\documentclass[11pt]{article}

\usepackage[T1]{fontenc}
\usepackage[utf8]{inputenc}

\usepackage{fancyhdr}

\usepackage[inline]{enumitem}
\usepackage{multicol,multirow,booktabs}

\usepackage{times}

\usepackage[hang,small,labelfont=bf,up,textfont=it,up]{caption}
\usepackage{subcaption}
\usepackage{amsmath,amsfonts,amsthm,amssymb,bm,dsfont}	
\setlist{nolistsep} 

\usepackage{graphicx}
\usepackage[a4paper,top=2.5cm,bottom=2.5cm,left=2.2cm,right=2.2cm]{geometry}

\usepackage{numprint}
\npthousandsep{,}\npthousandthpartsep{}\npdecimalsign{.}

\DeclareMathOperator*{\argmax}{argmax}

\newcommand{\dif}{\mathop{}\!\mathrm{d}}
\renewcommand{\vec}[1]{\boldsymbol{#1}}
\newcommand{\mvec}[1]{\mathbf{#1}}

\newcommand{\der}[1]{\frac{\partial}{\partial{#1}}}
\newcommand{\lolik}{\log L(\mathcal H_T;\bm\Psi)}
\newcommand{\lolikt}{\log L(\mathcal H_T;\tilde{\bm\Psi})}
\newcommand{\lolikfull}{\log L(\mathcal H_T;\tilde{\bm\Psi},\mvec B,\mvec Z)}

\newcommand{\titledoc}{Mutually exciting point process graphs \\ for modelling dynamic networks}
\newcommand{\titleshort}{Mutually exciting point process graphs for modelling dynamic networks}

\providecommand{\keywords}[1]{{\small{\textbf{\textit{Keywords ---}} #1}}}

\usepackage{natbib}
\usepackage{setspace}

\usepackage{thmtools}
\newtheorem{proposition}{Proposition}

\usepackage{authblk}

\usepackage{mathtools}
\mathtoolsset{showonlyrefs}

\usepackage[ruled,linesnumbered,noend]{algorithm2e}
\SetKwInput{KwInput}{Input}
\let\oldnl\nl
\newcommand{\nonl}{\renewcommand{\nl}{\let\nl\oldnl}}
\makeatother

\author{Francesco Sanna Passino}
\author{Nicholas A. Heard}

\affil{Department of Mathematics, Imperial College London \\ 180 Queen's Gate, SW7 2AZ, London}

\date{}

\title{\huge\textbf{\titledoc}}

\allowdisplaybreaks

\usepackage[colorlinks=true,linkcolor=blue,citecolor=blue,urlcolor=blue]{hyperref}
\numberwithin{equation}{section}

\begin{document}

\maketitle


\begin{abstract}
A new class of models for dynamic networks is proposed, called mutually exciting point process graphs (MEG). 
MEG is a scalable network-wide statistical model for point processes with dyadic marks, which can be used for anomaly detection when assessing the significance of future events, including previously unobserved connections between nodes. The model combines mutually exciting point processes to estimate dependencies between events and latent space models to infer relationships between the nodes. The intensity functions for each network edge are characterised exclusively by node-specific parameters, which allows information to be shared across the network. 
This construction enables estimation of intensities even for unobserved edges, which is particularly important in real world applications, such as computer networks arising in cyber-security.
A recursive form of the log-likelihood function for MEG is obtained, which 
is used to derive
fast inferential procedures 
via modern gradient ascent algorithms. 
An alternative EM algorithm is also derived. 
 The model and algorithms are tested on simulated graphs and real world 
 datasets, demonstrating excellent performance.
\end{abstract}


\section{Introduction} \label{intro}

Dynamic networks are encountered in many domains, representing, for example, interactions in social networks, messaging applications, or computer networks. Event data from dynamic networks are observed as triplets $(t_1,x_1,y_1), \ldots, (t_m,x_m,y_m)$, where $0\leq t_1\leq t_2\leq \ldots$ are event times and the dyadic marks $(x_k,y_k)$ denote the source and destination nodes, each belonging to a set of nodes $V=\{1,\dots,n\}$ of size $n$. 
The sequence of graph edges $(x_1,y_1),\dots,(x_m,y_m)$ induces a directed \textit{network adjacency matrix} $\mvec A=\{A_{ij}\}\in\{0,1\}^{n\times n}$ where $A_{ij}=1$ if node $i$ connected to node $j$ at least once during the entire observation period,
and $A_{ij}=0$ otherwise. 
This article presents a new class of models for the arrival of connection events between nodes in a network, called \textit{mutually exciting graphs} (MEG).
The MEG model builds upon \textit{mutually exciting point processes} and \textit{latent space models}.  

Mutually exciting point processes have been already successfully used for a variety of different applications: modelling of earthquakes \citep{Ogata88}, financial markets \citep{Bowsher07}, criminal activities \citep{Mohler11,Stomakhin11}, and popularity of tweets \citep{Zhao15,Chen18}. 
Let $t_1,t_2,\dots,t_m$ denote an increasing sequence of observed event times, and $N(t)=\sum_{k=1}^m \mathds 1_{[0,t]}(t_k)$ the corresponding counting process, representing the number of events observed up to time $t$. A counting process can be characterised by its \emph{conditional intensity function} $\lambda(t) = \lim_{\delta\to0}\mathbb E[N(t+\delta)-N(t)\vert\mathcal H_t]/\delta$, representing the expected rate of event times conditioned on the history $\mathcal H_t$ of the process up to time $t$. 
For self-exciting processes, the conditional intensity $\lambda(t)$ 
is assumed to depend on the last $r$ observed arrival times:
\begin{equation}
\lambda(t) = \lambda + \sum_{k>N(t)-r}^{N(t)} \omega(t-t_k), \label{mep}
\end{equation}
where $\lambda\in\mathbb R_+$ is a baseline intensity level and $\omega(\cdot)$ is a non-increasing and non-negative excitation function. For simplicity, $\omega(\cdot)$ is usually chosen to be a scaled exponential function: $\omega(t)=\beta\exp\{-(\beta+\theta)t\}$, where $\beta\geq 0$ and $\theta>0$. Usually, $\beta$ is referred to as \textit{jump} and $\beta+\theta$ as \textit{decay rate}. Alternative choices of $\omega(\cdot)$ are nonparametric step functions \citep{PriceWilliams19}, or the power-law $\omega(t)=\theta(t+\gamma)^{-1-\delta}$, where $\theta\geq 0$, $\beta,\delta>0$ and $\theta<\delta\beta^{\delta}$ \citep{Ozaki79}. In the literature, two extreme cases for the intensity in \eqref{mep} are usually considered: $r=1$, corresponding to a first order Markov-like structure, and $r=\infty$, called a Hawkes process \citep{Hawkes71}. 
Intuitively, if $r=1$, the intensity only depends on the time elapsed since the last event. 
On the other hand, if $r=\infty$, the conditional intensity depends on \emph{all} observed events, 
downweighted according to the elapsed time. 
If $r=0$, the model reduces to a simple Poisson process, such that all inter-arrival times are independent and exponentially distributed with rate $\lambda$.  

In large graphs, simultaneously modelling all the edge processes using individual intensities of the form \eqref{mep} is computationally challenging, and ignores possible correlations between different edges and nodes. Inference would require estimating $\mathcal O(n^2)$ parameters, or $\mathcal O\{\mathrm{nnz}(\mvec A)\}$ parameters if the graph is sparse, where $\mathrm{nnz}(\cdot)$ denotes the number of non-zero entries in a matrix. This is not feasible in most real-world applications. 
Furthermore, this approach would not parameterise \textit{new edges} appearing after the model training period. 
Hence, traditional statistical models for networks, for example latent space models \citep{Hoff02}, aim to reduce the representation of the network to $\mathcal O(n)$ parameters. Inspired by the literature on latent space models for network adjacency matrices, here a dynamic graph is modelled through the edge-specific point processes with intensity functions parametrised by node-specific latent features. 
In standard latent space network models, the probability of a link between two nodes is expressed as a function of node-specific latent vectors $\vec a_i,\vec b_j\in\mathbb R^d$, such that $\mathbb P(A_{ij}=1)=f(\vec a_i,\vec b_j)$, for some kernel function $f$.
In this work, it is assumed that the arrival times on each observed network edge can be modelled using a mutually exciting point process depending on node-specific characteristics. 

The related literature on mutually exciting point processes is vast, although mostly focusing on 
univariate and multivariate point processes; limited attention is devoted to using such processes for modelling large dynamic graphs. 
Hawkes processes are traditionally used to estimate causal interactions within multivariate processes \citep{Linderman14}, because of their appealing theoretical properties in terms of Granger causality and directed information~\citep{Etesami16,Eichler17}. Hawkes processes have also been used in \cite{Fox16} to analyse e-mail networks, primarily focusing on point processes on each node. \cite{Blundell12} proposed Hawkes processes to model reciprocating relationships between graph communities. \cite{Miscouridou18} extend this approach, proposing Hawkes process models for temporal interaction data with reciprocation, using compound completely random measures. 
The approach proposed in this paper is also related to \cite{Perry13}, who consider directed interactions within dynamic networks as a multivariate point process using a Cox multiplicative intensity model, with covariates depending on the history of the process. 
The MEG model proposed in this work is different from alternative methodologies proposed in the literature, since it uses mutually exciting processes at the edge level, parametrised only by node-specific features.

Furthermore, dynamic models for network snapshots observed at \textit{discrete} points in time have also been proposed in the literature. In particular, the methodology proposed in this work could be related to dynamic latent space models, which are based on latent feature representations of each node, evolving according to a temporal dynamics. Examples are \cite{Sarkar06,Krivitsky14,Sewell15,Durante16,Lee21}. The MEG model proposed in this article extends the latent feature framework to a continuous time setting, using node-specific latent vectors to parametrise point processes on each edge.

The remainder of this article is structured as follows: Section~\ref{meg} introduces the MEG model, followed by a description of the related inferential procedures in Section~\ref{sec:inference}. Section~\ref{simulation} discusses simulation from the model and the calculation of $p$-values for each network event. 
Results on simulated and real-world computer networks are discussed in Section~\ref{results}.

\section{Mutually exciting point process graphs} \label{meg}

The main contribution proposed in this article is a \textit{mutually exciting graph model} (MEG) for dynamic network point processes, defined by an $n\times n$ time-varying matrix of non-negative functions $\bm\lambda(t)=\{\lambda_{ij}(t)\}$. Each entry $\lambda_{ij}(t)$ is the conditional intensity of the counting process $N_{ij}(t)=\sum_{k=1}^m\mathds 1_{[0,t]\times\{i\}\times\{j\}}(t_k,x_k,y_k)$ of events occurring on the graph edge $(i,j)$, such that $\lambda_{ij}(t)=\lim_{\delta\to0}\mathbb E[N_{ij}(t+\delta)-N_{ij}(t)\vert\mathcal H_t]/\delta$. 
For generality, it is assumed that for each edge $(i,j)$ there exists a changepoint $\tau_{ij}\geq0$ after which the edge becomes observable. In the simplest case, $\tau_{ij}=0$ for all $i$ and $j$.

To parameterise $\bm\lambda(t)$, each entry is represented as an additive model with three non-negative components. The first, denoted $\alpha_i(t)$, characterises the process of arrival times involving $i$ as source node; the second, $\beta_j(t)$, corresponds to arrivals for which $j$ is the destination node; the third, $\gamma_{ij}(t)$, is an interaction term which will also be parameterised by node-specific parameters, giving: 
\begin{equation}
\lambda_{ij}(t) = 
\alpha_i(t)+\beta_j(t)+\gamma_{ij}(t), \quad t\geq\tau_{ij}. \label{megg}
\end{equation}
Note that the intensity function \eqref{megg} resembles the link function used in additive and multiplicative effect network models for network adjacency matrices, proposed in \cite{Hoff21}. 

Define the source and destination counting processes as $N_i(t)=\sum_{k=1}^m\mathds 1_{[0,t]\times\{i\}}(t_k,x_k)$ and $N_j^\prime(t)=\sum_{k=1}^m\mathds 1_{[0,t]\times\{j\}}(t_k,y_k)$. Furthermore, let $\ell_{i1},\ell_{i2},\dots$ denote the indices $\{k:x_k=i\}$ of the arrival times such that $i$ appears as source node, and $\ell_{j1}^\prime,\ell_{j2}^\prime,\dots$ denote the event indices $\{k:y_k=j\}$ for which $j$ is the destination node. To allow self excitation of both source and destination nodes, the latent functions $\alpha_i(t)$ and $\beta_j(t)$ are assigned a similar form to the conditional intensity \eqref{mep}: 
\begin{align}
\alpha_i(t)=\alpha_i+ \sum_{k>N_i(t)-r}^{N_i(t)} \omega_i(t-t_{\ell_{ik}}), \quad 
\beta_j(t)=\beta_j + \sum_{k>N^\prime_j(t)-r}^{N^\prime_j(t)}  \omega_j^\prime(t-t_{\ell_{jk}^\prime}), \label{node_process}
\end{align}
where $\bm\alpha=(\alpha_1,\ldots,\alpha_n),\bm\beta=(\beta_1,\ldots,\beta_n)\in\mathbb R_+^n$ are node-specific baseline intensity levels, and $\omega_i,\omega_i^\prime$ 
are node-specific, non-increasing excitation functions from $\mathbb R_+$ to $\mathbb R_+$. For simplicity, the excitation functions 
assume the following scaled exponential form, for non-negative parameters $\bm\mu_i,\bm\mu_j^\prime,\bm\phi_i,\bm\phi_j^\prime\in\mathbb R_+^n$:
\begin{align}
\omega_i(t) = \mu_i\exp\{-(\mu_i+\phi_i)t\}, \quad \omega_j^\prime(t) = \mu_j^\prime\exp\{-(\mu_j^\prime+\phi_j^\prime) t\}. \label{scaled_expo}
\end{align}
Scaled exponential excitation functions have significant computational advantages for inference in MEG models: in particular, the log-likelihood can be expressed in recursive form and evaluated in linear time, which speeds up inference and allows the methodology to scale to large graphs. These aspects will be more extensively discussed in Section~\ref{sec:rec_lik}.

Similarly, 
let $\ell_{ij1},\ell_{ij2},\dots$ be the indices $\{k:x_k=i,y_k=j\}$ of the events observed on the edge $(i,j)$. The interaction term $\gamma_{ij}(t)$ in \eqref{megg} assumes a similar form to \eqref{node_process}, but with a background rate obtained as the inner product between two node-specific $d$-dimensional baseline parameter vectors $\bm\gamma_i,\bm\gamma^\prime_j \in \mathbb{R}^d_+,\ d\in\mathbb N$:
\begin{equation}
  \gamma_{ij}(t)= \bm\gamma_i^\intercal \bm\gamma^\prime_j 
  + \sum_{k>N_{ij}(t)-r}^{N_{ij}(t)} \omega_{ij}(t-t_{\ell_{ijk}}). \label{edge_process}
\end{equation}
The excitation function $\omega_{ij}(t)$ is also expressed as a sum of scaled exponential functions, parameterised by four node-specific, non-negative latent $d$-vectors $\bm\nu_i,\bm\nu^\prime_j,\bm\theta_i,\bm\theta^\prime_j\in\mathbb{R}_+^d$:
\begin{equation}
\omega_{ij}(t) = \sum_{\ell=1}^d \nu_{i\ell}\nu_{j\ell}^\prime\exp\{-(\theta_{i\ell}+\nu_{i\ell})(\theta_{j\ell}^\prime+\nu_{j\ell}^\prime) t\}.\label{eq:omega_ij}
\end{equation}
The inner product baseline and products within the scaled exponential excitation functions are inspired by random dot product graph models \citep[see, for example,][]{Athreya18} for link probabilities. This choice is helpful to obtain closed form expression for inference in MEG models, as discussed in Section~\ref{sec:em}. Alternative options, inspired by other latent space models, could also be used for the baseline, such as $\|\bm\gamma_i-\bm\gamma_j^\prime\|_2$ \citep{Hoff02}.

A cartoon example of the intensity $\lambda_{ij}(t)$ for the $d=1$ dimensional MEG model with scaled exponential functions is given in Figure~\ref{simulated_process}, with $\alpha_i=0.2,\ \mu_i=0.5,\ \phi_i=0.5,\ \beta_j=0.1,\ \mu_j^\prime=0.8,\ \phi_j^\prime=0.2,\ \gamma_i=0.8,\ \nu_i=0.9,\ \theta_i=1.1,\ \gamma_j^\prime=0.6,\ \nu_j^\prime=0.3,\ \theta_j^\prime=0.2$. In Figure~\ref{sim_all}, the edge intensity function jumps at each event time involving source node $i$ or destination node $j$, or both. In particular, larger jumps in $\lambda_{ij}(t)$, of size $\mu_i+\mu_j^\prime+\nu_i\nu_j^\prime$ for $r=\infty$, are observed when events are observed on the edge $(i,j)$ (triangles). The intensity also increases if events are observed from source node $i$ (circles, \textit{cf.} Figure~\ref{sim_alpha}) or to destination node $j$ (squares, \textit{cf.} Figure~\ref{sim_beta}), with jumps of size $\mu_i$ and $\mu_j^\prime$ respectively for $r=\infty$. For $r=1$, the intensity $\lambda_{ij}(t)$ is bounded by construction at $\alpha_i+\beta_j+\gamma_i\gamma_j^\prime+\mu_i+\mu_j^\prime+\nu_i\nu_j^\prime$.


\begin{figure}[t]
\centering
\begin{subfigure}[t]{.495\textwidth}
\centering
\caption{$\alpha_i(t)$}
\includegraphics[width=0.975\textwidth]{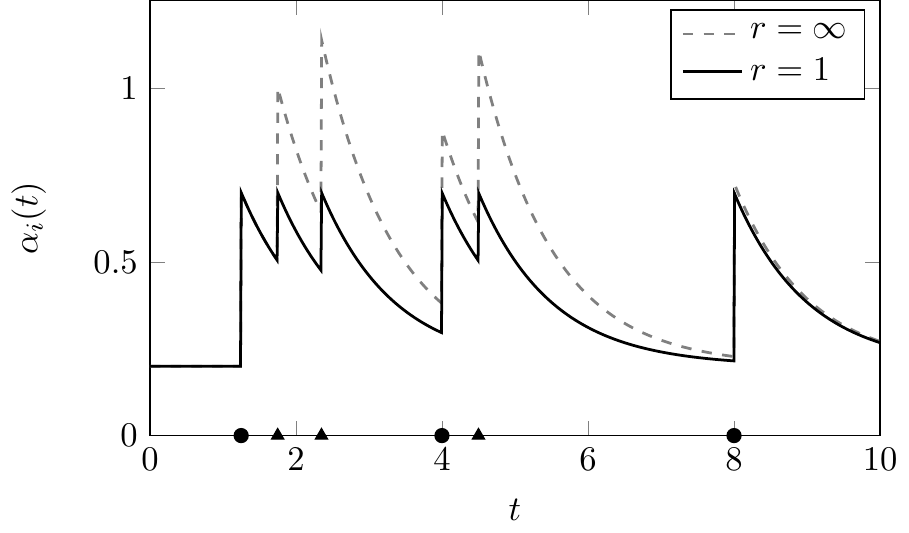}
\label{sim_alpha}
\end{subfigure}
\begin{subfigure}[t]{0.495\textwidth}
\centering
\caption{$\beta_i(t)$}
\includegraphics[width=0.975\textwidth]{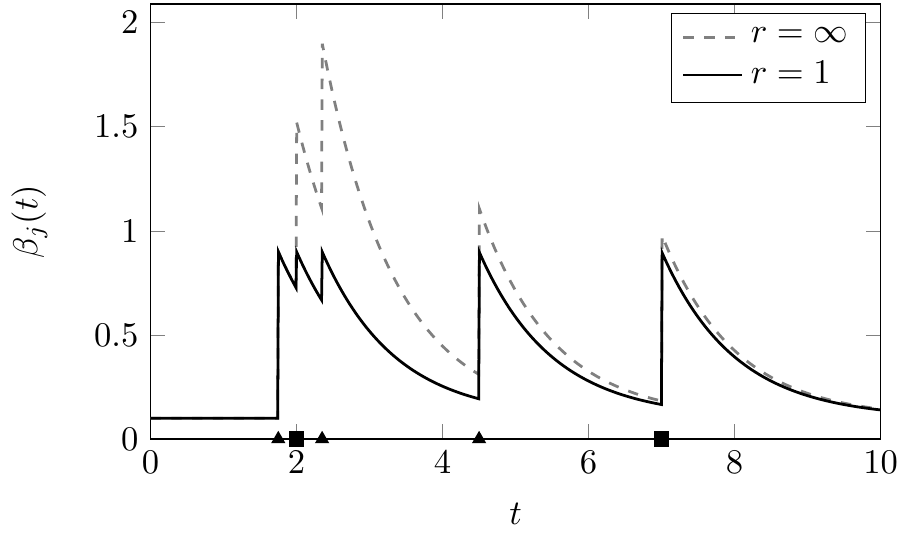}
\label{sim_beta}
\end{subfigure}
\begin{subfigure}[t]{0.495\textwidth}
\centering
\caption{$\gamma_{ij}(t)$}
\includegraphics[width=0.975\textwidth]{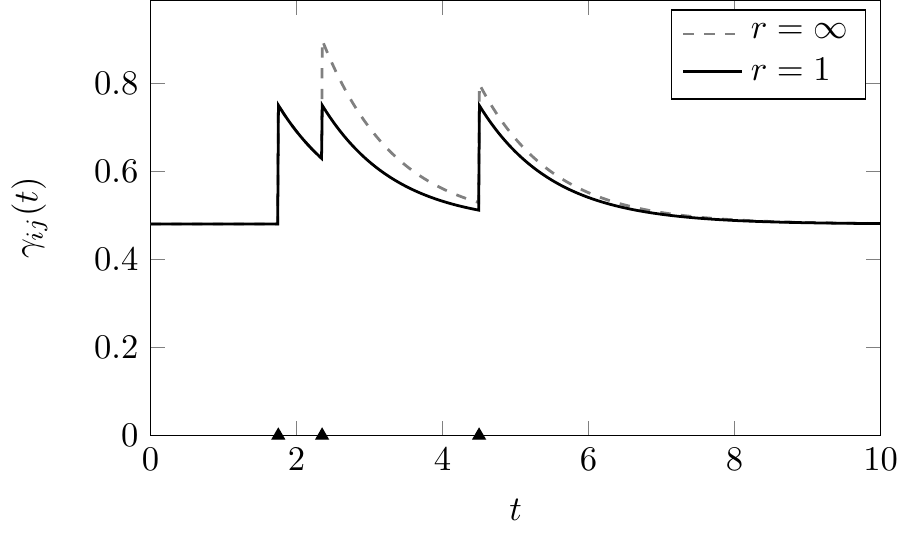}
\label{sim_gamma}
\end{subfigure}
\begin{subfigure}[t]{0.495\textwidth}
\centering
\caption{$\lambda_{ij}(t)$}
\includegraphics[width=0.975\textwidth]{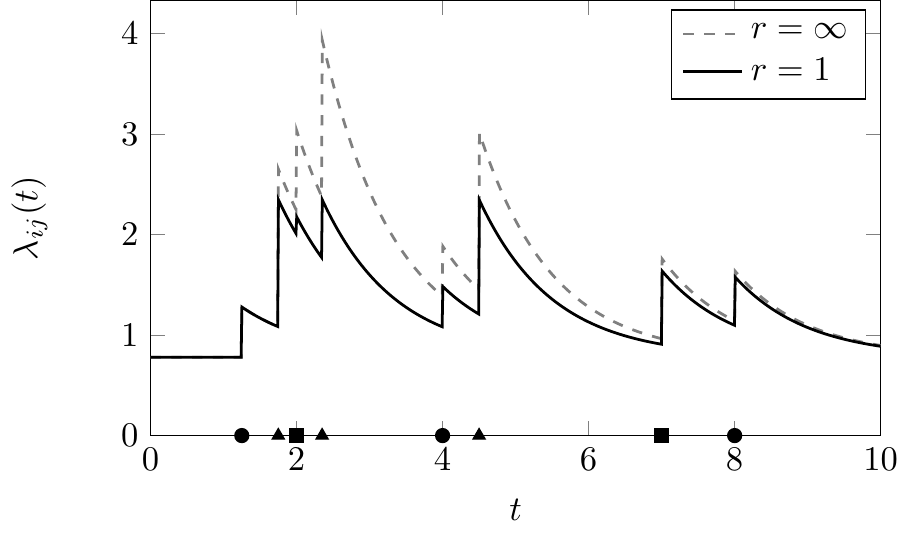}
\label{sim_all}
\end{subfigure}
\caption{Cartoon of a 1-dimensional MEG model \eqref{megg}--\eqref{eq:omega_ij} for $r=1$ and $r=\infty$. Event times are marked on the $x$-axis. Events with source node $i$ and destination node $j$ are denoted by triangles ($\blacktriangle$); other events with source node $i$ are are denoted with circles ($\bullet$), and other events with destination node $j$ are denoted by squares ($\blacksquare$). 
For each event time, a jump in the corresponding intensities is observed.  
}
\label{simulated_process}
\end{figure}

A key feature of the model is the representation of the intensity \eqref{megg} with only node-specific parameters $\bm\Psi=(\bm\alpha,\bm\mu,\bm\phi,\bm\beta,\bm\mu^\prime,\bm\phi^\prime,\bm\gamma,\bm\nu,\bm\theta,\bm\gamma^\prime,\bm\nu^\prime,\bm\theta^\prime)$. This construction allows estimation of intensities even for unobserved edges. Therefore, in practical applications, when a new link is observed, it is possible to immediately provide an estimate of the intensity of the process on that edge. This is a substantial difference with respect to models based on edge-specific parameters for the edge intensities. Such models do not perform well for scoring new links, since the score for a new observation could only be based on a prior guess of the intensity, whereas MEG ``borrows strength'' from events observed on similar nodes and edges in the graph, providing an informed estimate of the intensity function on the newly observed edge. 

For a sequence of observed events $\mathcal H_T=\{(x_1,y_1,t_1),\dots,(x_m,y_m,t_m)\}$, with event times in $[0,T]$, the log-likelihood \citep{Daley02} of a generic MEG model is:
\begin{equation}
\lolik = \sum_{i=1}^n\sum_{j=1}^n\left\{\sum_{k=1}^{n_{ij}}\log
\lambda_{ij}(t_{\ell_{ijk}})-\int_{\tau_{ij}}^T \lambda_{ij}(t)\mathrm \dif t\right\}. \label{loglik}
\end{equation}
where $n_{ij}$ is the number of events observed on edge $(i,j)$. 
Explicit forms of the likelihood function \eqref{loglik} 
with $d=1$ and $r=1$ or $r=\infty$, are 
presented in the Appendix. 

\subsection{Computational issues with the calculation of the likelihood} \label{sec:rec_lik}

For $r=\infty$, the main computational burden associated with the calculation of the log-likelihood \eqref{loglik} is the double summation over each $(i,j)$ pair of the sum of intensities $\lambda_{ij}(t_{\ell_{ijk}})$ for the events $t_{\ell_{ijk}}$, and the summations required in \eqref{node_process} and \eqref{edge_process} to evaluate that intensity for each event.
This section discusses a recursive form of the log-likelihood \eqref{loglik} for the MEG model with $r=\infty$, which can be evaluated in linear time on each active edge, significantly reducing the computational requirements. This is a significant computational advantage of scaled exponential excitation functions, which makes them particularly appealing for practical applications.
Assume sequences of arrival times $t_{i1}<\dots<t_{iN_i(T)}$ involving $i$ as source node, and $t_{j1}^\prime<\dots<t_{jN_j^\prime(T)}^\prime$ such that $j$ is the destination of the connection. Within each pair of sequences, assume that a subset of $n_{ij}\leq\min\{N_i(T),N_j^\prime(T)\}$ events is observed on the edge $(i,j)$, and denote the indices of such events as $u_{ij,1},\dots,u_{ij,n_{ij}}$ and $u_{ij,1}^\prime,\dots,u_{ij,n_{ij}}^\prime$. 
The terms in the first summation in the log-likelihood \eqref{loglik} can then be written as:
\begin{align}
\log\lambda_{ij}(t_{\ell_{ijk}})  = 
\log\bigg\{ & \alpha_i + \mu_i\sum_{h=1}^{u_{ij,k}-1} e^{-(\mu_i+\phi_i)(t_{\ell_{ijk}}-t_{ih})} + \beta_j + \mu_j^\prime\sum_{h=1}^{u_{ij,k}^\prime-1} e^{-(\mu_j^\prime+\phi_j^\prime)(t_{\ell_{ijk}}-t_{jh}^\prime)} \notag \\ 
& + \bm\gamma_i^\intercal\bm\gamma_j^\prime +\sum_{q=1}^d\nu_{iq}\nu_{jq}^\prime\sum_{h=1}^{k-1} e^{-(\nu_{iq}+\theta_{iq})(\theta_{jq}^\prime+\nu_{jq}^\prime)(t_{\ell_{ijk}}-t_{\ell_{ijh}})}\bigg\}. \label{mega_lik}
\end{align}
Using a technique similar to the method proposed in \cite{Ogata78}, it is possible to calculate \eqref{mega_lik} in linear time using a recursive formulation of the inner summations. For $k\in\{1,2,\dots,n_{ij}\}$, define $\psi_{ij}(k)$, $\psi_{ij}^\prime(k)$ and $\tilde\psi_{ijq}(k)$ as follows:
\begin{gather}
\psi_{ij}(k) = \sum_{h=1}^{u_{ij,k}-1} e^{-(\mu_i+\phi_i)(t_{\ell_{ijk}}-t_{ih})}, 
\hspace{2cm} \psi_{ij}^\prime(k) = \sum_{h=1}^{u_{ij,k}^\prime-1} e^{-(\mu_j^\prime+\phi_j^\prime)(t_{\ell_{ijk}}-t_{jh}^\prime)}, \\ 
\tilde\psi_{ijq}(k)=\sum_{h=1}^{k-1} e^{-(\nu_{iq}+\theta_{iq})(\nu_{jq}^\prime+\theta_{jq}^\prime)(t_{\ell_{ijk}}-t_{\ell_{ijh}})},\ q=1,\dots,d. \label{psi_def}
\end{gather}
Using \eqref{mega_lik} and \eqref{psi_def}, the first term of the log-likelihood \eqref{loglik} becomes:
\begin{equation}
\sum_{k=1}^{n_{ij}} \log\lambda_{ij}(t_{\ell_k}) = \sum_{k=1}^{n_{ij}} \log\left\{\alpha_i+\beta_j+\gamma_i\gamma_j^\prime + \mu_i\psi_{ij}(k) + \mu_j^\prime\psi_{ij}^\prime(k) + \sum_{q=1}^d\nu_{iq}\nu_{jq}^\prime\tilde\psi_{ijq}(k) \right\}. \label{meg_lik2}
\end{equation}
The expression can be evaluated in linear time using the recursive equations for $\psi_{ij}(k), \psi_{ij}^\prime(k)$ and $\tilde\psi_{ijq}(k)$ presented in the following proposition, proved in the Appendix.  

\begin{proposition}
The terms $\psi_{ij}(k), \psi_{ij}^\prime(k)$ and $\tilde\psi_{ij}(k)$ can be written recursively as follows:
\begin{align}
& \psi_{ij}(k) = e^{-(\mu_i+\phi_i)(t_{\ell_{ijk}}-t_{\ell_{ij,k-1}})}\left[1+\psi_{ij}(k-1)\right] + \sum_{h=u_{ij,k-1}+1}^{u_{ij,k}-1} e^{-(\mu_i+\phi_i)(t_{\ell_{ijk}}-t_{ih})}, \label{psi}\\
& \psi_{ij}^\prime(k) = e^{-(\mu_j^\prime+\phi_j^\prime)(t^\prime_{\ell_k^\prime}-t^\prime_{\ell_{k-1}^\prime})}\left[1+\psi_{ij}^\prime(k-1)\right] + \sum_{h=u_{ij,k-1}^\prime+1}^{u_{ij,k}^\prime-1} e^{-(\mu_j^\prime+\phi_j^\prime)(t_{\ell_{ijk}}-t_{jh}^\prime)},\\ 
& \tilde\psi_{ijq}(k) = e^{-(\nu_{iq}+\theta_{iq})(\nu_{jq}^\prime+\theta_{jq}^\prime)(t_{\ell_{ijk}}-t_{\ell_{ij,k-1}})}\left[1+\tilde\psi_{ijq}(k-1)\right].
\label{psi_prime}
\end{align}
\end{proposition}

\subsection{Extension to undirected and bipartite graphs}

The proposed modelling framework has been presented for directed graphs, but could be easily extended to undirected and bipartite networks. For undirected graphs $\mvec A=\mvec A^\intercal$,
hence there is no distinction between source and destination nodes. Therefore, $\beta_j(t)$ in \eqref{megg} could be simply be replaced by $\alpha_j(t)$, and $\gamma_{ij}(t)$ modified as follows:
\begin{equation}
\gamma_{ij}(t) = \bm\gamma_i^\intercal \bm\gamma_j + \sum_{k>N_{ij}(t)-r}^{N_{ij}(t)}\sum_{\ell=1}^d \nu_{i\ell}\nu_{j\ell}\exp\{-(\theta_{i\ell}+\nu_{i\ell})(\theta_{j\ell}+\nu_{j\ell}) t\}.
\end{equation}
Furthermore, bipartite graphs can be considered as special cases of directed graphs, where the node set $V=V_1\cup V_2$ is divided into two sets $V_1$ and $V_2$ of cardinality $n_1$ and $n_2$, such that $V_1\cap V_2=\varnothing$, and all the edges are of the form $(i,j)$ with $i\in V_1$ and $j\in V_2$. Therefore, the intensity function \eqref{megg} and the corresponding components \eqref{node_process} and \eqref{edge_process} 
still hold. 

\section{Inference via maximum likelihood estimation} \label{sec:inference}

Inference in Hawkes processes is usually carried out using maximum likelihood estimation (MLE) via the EM algorithm or gradient ascent methods, since it is not possible to optimise the likelihood analytically. Similar issues arise for the log-likelihood \eqref{loglik} for MEG models.
Only a small subset of the parameters has a closed-form solution for the MLE: the start times $\tau_{ij}$. 
If $\tau_{ij}$ in \eqref{loglik} is unknown, the maximum likelihood estimates is simply $\hat\tau_{ij}=t_{\ell_{ij1}}$ if at least one event is observed on the edge, and $\hat\tau_{ij}=\infty$ otherwise. 
Intuitively, this is reasonable: the best guess about the start time of activity on an edge simply corresponds to the first observation on that edge. 
A formal proof of the result is provided in the Appendix. 
For maximising \eqref{loglik} with respect to the remaining parameters $\bm\Psi$, two strategies are deployed: 
the Expectation-Maximisation algorithm \citep[EM,][]{Dempster77}, and 
the adaptive moment estimation method \citep[\textit{Adam,}][]{Kingma15}. 

Note that issues with MLE might arise when the parameters lie at the boundaries of the parameter space. For example, for a non-negative finite jump $0<\mu_i<\infty$, an infinitely fast decaying rate $\phi_i\to\infty$ would make the resulting process 
a simple Poisson process for $\alpha_i(t)$ in \eqref{meg}. Similarly, if $0<\phi_i<\infty$, a jump size $\mu_i\to0$ would make $\alpha_i(t)$ again correspond to a Poisson process with rate $\alpha_i$. Similar considerations can be made about the parameters of the excitation functions of the remaining components in \eqref{meg}.
Additionally, identifiability issues are observed if further constraints are not imposed on the parameters: for example, subtracting a constant $c\in[0,\min\{\alpha_1,\dots,\alpha_n\})$ for all $\alpha_i$, and adding the \textit{same} constant to all $\beta_j$, returns the same log-likelihood function \eqref{loglik}. For identifiability, at least one value of $\alpha_i$ or $\beta_j$ must be kept fixed.
Identifiability issues also arise from the interaction term: the inner product $\bm\gamma_i^\intercal\bm\gamma_j^\prime$ is invariant to orthogonal transformations of $\bm\gamma_i$ and $\bm\gamma_j^\prime$ preserving non-negativity of the vectors.
Similarly, for a constant $c\in\mathbb R_+$, the interaction parameters $\bm\nu_i, \bm\theta_i, \bm\nu_j^\prime$ and $\bm\theta_j^\prime$ produce the same excitation function as $c\bm\nu_i, c\bm\theta_i, \bm\nu_j^\prime/c$ and $\bm\theta_j^\prime/c$. This leads to highly multimodal log-likelihood functions, and to a non-unique MLE.
This problem is inconsequential for prediction, since the predictive distribution of new events depends upon a function of the parameters which is identifiable. Similarly, for assessing robustness of parameter estimation procedures, identifiable transformations of the parameters can be considered: for example, the sums $\alpha_i+\beta_j$ in the main effects model are identifiable, or the products $(\nu_i+\theta_i)(\nu_j^\prime + \theta_j^\prime)$ for $d=1$. An example will be given in Section~\ref{sim_full}.

\subsection{Inference via the EM algorithm} \label{sec:em}

An EM algorithm can be conveniently implemented following a network-wide extension of the procedure of \cite{Fox16}, after adopting a simple reparametrisation of the log-likelihood \eqref{loglik}. In particular, the scaled exponential decay rates $\mu_i+\phi_i, \mu_j^\prime+\phi_j^\prime,  \nu_{iq}+\theta_{iq}$ and $\nu_{jq}^\prime+\theta_{jq}^\prime$ are rewritten as $\tilde\phi_i, \tilde\phi_j^\prime, \tilde\theta_{iq}$ and $\tilde\theta_{jq}^\prime$, where $\tilde\phi_i>\mu_i, \tilde\phi_j^\prime>\mu_j^\prime, \tilde\theta_{iq}>\nu_i$, and $\tilde\theta_{jq}^\prime>\nu_j^\prime$. 
Similarly, the jumps $\mu_i, \mu_j^\prime, \nu_{iq}$ and $\nu_{jq}^\prime$ are expressed as the product between the decay rates $\tilde\phi_i, \tilde\phi_j^\prime, \tilde\theta_{iq}$ and $\tilde\theta_{jq}^\prime$ and the ratios 
between the jump and decay rates, denoted:
\begin{align}
\tilde\mu_i = \frac{\mu_i}{\mu_i+\phi_i}, && \tilde\mu_j^\prime = \frac{\mu_j^\prime}{\mu_j^\prime+\phi_j^\prime},&& \tilde\nu_{iq}=\frac{\nu_{iq}}{\nu_{iq}+\theta_{iq}},&& \tilde\nu_{jq}^\prime=\frac{\nu_{jq}^\prime}{\nu_{jq}^\prime+\theta_{jq}^\prime},
\end{align}
where such parameters lie in $[0,1]$.
For example, under the two equivalent parametrisations, $\omega_i(t) = \mu_i\exp\{-(\mu_i+\phi_i)t\} = \tilde\mu_i\tilde\phi_i\exp\{-\tilde\phi_it\}$.
The vector of all parameters can then be equivalently rewritten as $\tilde{\bm\Psi}=(\bm\alpha,\tilde{\bm\mu},\tilde{\bm\phi},\bm\beta,\tilde{\bm\mu}^\prime,\tilde{\bm\phi}^\prime,\bm\gamma,\tilde{\bm\nu},\tilde{\bm\theta},\bm\gamma^\prime,\tilde{\bm\nu}^\prime,\tilde{\bm\theta}^\prime)$, using the updated notation. 
Furthermore, consider the sequence of arrival times $t_{i1}<\cdots<t_{iN_i(T)}$ involving $i$ as source node, and $t_{j1}^\prime<\cdots<t_{jN_j^\prime(T)}^\prime$ such that $j$ is the destination of the connection. Similarly, let the sequence $t_{ij1}<\cdots<t_{ij{N_{ij}(T)}}$ denote the events on the edge $(i,j)$.
Using this revised notation, the conditional intensity function \eqref{megg} for an edge, for $t\geq \tau_{ij}$, is:
\begin{equation}
\lambda_{ij}(t)=\alpha_i+ \sum_{k>N_i(t)-r}^{N_i(t)} \omega_i(t-t_{ik}) + \beta_j + \sum_{k>N^\prime_j(t)-r}^{N^\prime_j(t)}  \omega_j^\prime(t-t_{jk}^\prime) + \sum_{q=1}^d \gamma_{iq}\gamma^\prime_{jq} 
  + \sum_{k>N_{ij}(t)-r}^{N_{ij}(t)}\sum_{q=1}^d \omega_{ijq}(t-t_{ijk}), \label{cif}
\end{equation}
where the excitation function $\omega_{ij}(\cdot)$ in \eqref{edge_process} has been expressed as a sum of $d$ functions $\omega_{ijq}:\mathbb R_+\to\mathbb R_+$, where $\omega_{ijq}(t)=\tilde\nu_{iq}\tilde\theta_{iq}\tilde\nu_{jq}^\prime\tilde\theta_{jq}^\prime\exp\{-\tilde\theta_{iq}\tilde\theta_{jq}^\prime t\}$ from \eqref{eq:omega_ij}.
Therefore, conditional on $t$, the subsequent event on the edge $(i,j)$ could be interpreted as the offspring of one of the $2+d+\min\{r,N_i(t)\}+\min\{r,N_j^\prime(t)\}+d\min\{r,N_{ij}(t)\}$  components of the intensity \eqref{cif}, each corresponding to a non-homogeneous Poisson process in $(t,\infty)$. In other words, $\lambda_{ij}(t)$ is written as a superimposition of conditional intensities of different processes, where the event allocations are missing data, 
giving a \emph{branching structure} to the event hierarchy.

For missing data problems, the traditional approach in statistics is to deploy the EM algorithm, 
which in this setting requires to introduce latent binary variables to reconstruct the branching structure. 
For events generated from the background rates $\alpha_i$, $\beta_j$ and $\gamma_{iq}\gamma^\prime_{jq},\ q=1,\dots,d$ (also known as \emph{immigrant events} in the literature), the corresponding latent variables are denoted by the letter $b$. In particular $b_{ij\ell}^{(\alpha)}\in\{0,1\}$ equals 1 if $t_{ij\ell}$ is a background event obtained from the Poisson process with rate $\alpha_i$, and 0 otherwise. Similarly, $b_{ij\ell}^{(\beta)}$ and $b_{ij\ell q}^{(\gamma)}$ denote whether the event $t_{ij\ell}$ is a background event from Poisson processes with rates $\beta_j$ and $\gamma_{iq}\gamma^\prime_{jq}$ respectively.
On the other hand, for the events that are not generated from the background rates, the corresponding latent variables are denoted with the letter $z$. In particular,
$z_{ij\ell k}^{(\alpha)}=1$ if $t_{ij\ell}$ is offspring of the $k$-th event such that node $i$ is source, and $0$ otherwise; a similar reasoning applies for $z_{ij\ell k}^{(\beta)}$, which instead considers the sequence of events such that node $j$ is destination. As before, it is necessary to introduce a further subscript for the interaction term: $z_{ij\ell kq}^{(\gamma)}=1$ if $t_{ij\ell}$ is offspring of the $k$-th event on the edge $(i,j)$, from the $q$-th additive component of the intensity.
If such latent variables are known, it is possible to write in simple form the \emph{complete data} log-likelihood, which also includes the information about the branching structure:
\begin{multline}
\lolikfull = \sum_{i=1}^n \sum_{j=1}^n \Bigg\{\sum_{\ell=1}^{n_{ij}} \Bigg[b_{ij\ell}^{(\alpha)}\log(\alpha_i) + \sum_{k>N_i(t_{ij\ell})-r}^{N_i(t_{ij\ell})} z_{ij\ell k}^{(\alpha)}[\log(\tilde\mu_i\tilde\phi_i)-\tilde\phi_i(t_{ij\ell}-t_{ik})] \\
+ b_{ij\ell}^{(\beta)}\log(\beta_i) + \sum_{k>N_j^\prime(t_{ij\ell})-r}^{N_j^\prime(t_{ij\ell})} z_{ij\ell k}^{(\beta)}[\log(\tilde\mu_j^\prime\tilde\phi_j^\prime)-\tilde\phi_j^\prime(t_{ij\ell}-t_{jk}^\prime)]
+\sum_{q=1}^d \Bigg(b_{ij\ell q}^{(\gamma)}[\log(\gamma_{iq})+\log(\gamma_{jq}^\prime)] \\ 
+ \sum_{k>N_{ij}(t_{ij\ell})-r}^{N_{ij}(t_{ij\ell})} z_{ij\ell kq}^{(\gamma)}[\log(\tilde\nu_{iq}\tilde\theta_{iq})+\log(\tilde\nu_{jq}^\prime\tilde\theta_{jq}^\prime)-\tilde\theta_{iq}\tilde\theta_{jq}^\prime(t_{ij\ell}-t_{ijk})]\Bigg)
\Bigg] -\int_{\tau_{ij}}^T \lambda_{ij}(t)\mathrm \dif t \Bigg\}. \label{lolik_full}
\end{multline}
The E-step of the EM algorithm consists in calculating $\mathbb E_{\mvec B,\mvec Z\vert\mathcal H_T,\bm\Psi^\ast}\{\lolikfull\}$, the expected value of the complete data log-likelihood \eqref{lolik_full} with respect to the distribution of the latent indicators $\mvec B$ and $\mvec Z$,  conditional on the observations $\mathcal H_T$ and parameter values $\tilde{\bm\Psi}^\ast$. From \eqref{lolik_full}, this reduces to calculating:
\begin{align}
\xi_\cdot^{(\cdot)} = 
\mathbb P_{\mvec B,\mvec Z\vert\mathcal H_T,\tilde{\bm\Psi}^\ast}\left\{b_\cdot^{(\cdot)}=1\ \big\vert\ \tilde{\bm\Psi}^\ast\right\}, 
& & 
\zeta_\cdot^{(\cdot)} = 
\mathbb P_{\mvec B,\mvec Z\vert\mathcal H_T,\tilde{\bm\Psi}^\ast}\left\{z_\cdot^{(\cdot)}=1\ \big\vert\ \tilde{\bm\Psi}^\ast\right\},
\end{align}
known as \textit{responsibilities}. Such 
probabilities are simply represented by the relative contributions of different components to the conditional intensity \eqref{cif}:
\begin{align}
& \xi_{ij\ell}^{(\alpha)} \propto {\alpha_i}, & & \zeta_{ij\ell k}^{(\alpha)} \propto {\tilde\mu_i\tilde\phi_i\exp\{-\tilde\phi_i(t_{ij\ell}-t_{ik})\}} \mathds 1_{(t_{ik},\infty)}(t_{ij\ell}), \\
& \xi_{ij\ell}^{(\beta)} \propto {\beta_j}, & & \zeta_{ij\ell k}^{(\beta)} \propto {\tilde\mu_j^\prime\tilde\phi_j^\prime\exp\{-\tilde\phi_j^\prime(t_{ij\ell}-t_{jk}^\prime)\}} \mathds 1_{(t_{jk}^\prime,\infty)}(t_{ij\ell}), \\
& \xi_{ij\ell q}^{(\gamma)} \propto {\gamma_{iq}\gamma_{jq}^\prime}, & & \zeta_{ij\ell kq}^{(\gamma)} \propto {\tilde\nu_{iq}\tilde\theta_{iq}\tilde\nu_{jq}^\prime\tilde\theta_{jq}^\prime\exp\{-\tilde\theta_{iq}\tilde\theta_{jq}^\prime(t_{ij\ell}-t_{ijk})\}} \mathds 1_{(t_{ijk},\infty)}(t_{ij\ell}),  \label{respo}
\end{align}
with normalising constant $\lambda_{ij}(t_{ij\ell})$, \textit{cf.} \eqref{meg} and \eqref{cif}, calculated using parameter values $\tilde{\bm\Psi}^\ast$.

At the M-step, the expectation $\mathbb E_{\mvec B,\mvec Z\vert\mathcal H_T,\tilde{\bm\Psi}^\ast}\{\lolikfull\}$ calculated at the E-step is maximised with respect to $\tilde{\bm\Psi}$, and updated parameter estimates are obtained. For most of the parameters in the MEG model with scaled exponential excitation function, the maxima are analytically available, and their form depends on the choice of $r$. For $r=\infty$:
\begingroup
\allowdisplaybreaks
\begin{align}
& \hat\alpha_i = \frac{\sum_{j=1}^n \sum_{\ell=1}^{n_{ij}} \xi_{ij\ell}^{(\alpha)}}{\sum_{j=1}^n (T-\min\{T,\tau_{ij}\})}, \notag
&& \hat{\tilde\mu}_i = \frac{\sum_{j=1}^n\sum_{\ell=1}^{n_{ij}}\sum_{k=1}^{N_i(t_{ij\ell})} \zeta_{ij\ell k}^{(\alpha)}}{\sum_{j=1}^n\sum_{k=1}^{n_i} [e^{-\tilde\phi_i\min\{T,\max\{\tau_{ij}-t_{ik},0\}\}} - e^{-\tilde\phi_i(T-t_{ik})}]}, \\
& \hat\beta_j = \frac{\sum_{i=1}^n \sum_{\ell=1}^{n_{ij}} \xi_{ij\ell}^{(\beta)}}{\sum_{i=1}^n (T-\min\{T,\tau_{ij}\})}, 
&& \hat{\tilde\mu}_j^\prime = \frac{\sum_{i=1}^n\sum_{\ell=1}^{n_{ij}}\sum_{k=1}^{N_j^\prime(t_{ij\ell})} \zeta_{ij\ell k}^{(\beta)}}{\sum_{i=1}^n\sum_{k=1}^{n_j^\prime} [e^{-\tilde\phi_j^\prime\min\{T,\max\{\tau_{ij}-t_{jk}^\prime,0\}\}} - e^{-\tilde\phi_j^\prime(T-t_{jk}^\prime)}]}, \notag \\
& \hat\gamma_{iq} = \frac{\sum_{j=1}^n \sum_{\ell=1}^{n_{ij}} \xi_{ij\ell q}^{(\gamma)}}{\sum_{j=1}^n \gamma_{jq}^\prime(T-\min\{T,\tau_{ij}\})},
&& \hat{\tilde\nu}_{iq} = \frac{\sum_{j=1}^n\sum_{\ell=1}^{n_{ij}}\sum_{k=1}^{N_{ij}(t_{ij\ell})} \zeta_{ij\ell kq}^{(\gamma)}}{\sum_{j=1}^n \tilde\nu_{jq}^\prime\sum_{k=1}^{n_{ij}} [1
- e^{-\tilde\theta_{iq}\tilde\theta_{jq}^\prime(T-t_{ijk})}]}, \label{est1} 
\end{align}
and similarly for $\hat\gamma_{jq}^\prime$ and $\hat{\tilde\nu}_{jq}^\prime$.
For the remaining parameters, an exact solution is not available, but recursive equations can be obtained. For example, again for $r=\infty$: 
\begin{gather}
\tilde\phi_i = \frac{\sum_{j=1}^n\sum_{\ell=1}^{n_{ij}}\sum_{k=1}^{N_i(t_{ij\ell})} \zeta_{ij\ell k}^{(\alpha)}}{\sum_{j=1}^n\{\sum_{\ell=1}^{n_{ij}}\sum_{k=1}^{N_i(t_{ij\ell})} \zeta_{ij\ell k}^{(\alpha)}(t_{ij\ell} - t_{ik}) + \tilde\mu_i\sum_{k=1}^{n_i} [(T-t_{ik})e^{-\tilde\phi_i(T-t_{ik})} - \tau_{ijk}^+e^{-\tilde\phi_i\tau_{ijk}^+}]\}}, \\
\tilde\theta_{iq} = \frac{\sum_{j=1}^n\sum_{\ell=1}^{n_{ij}}\sum_{k=1}^{N(t_{ij\ell})} \zeta_{ij\ell kq}^{(\gamma)}}{\sum_{j=1}^n\sum_{\ell=1}^{n_{ij}}\{ \sum_{k=1}^{N_{ij}(t_{ij\ell})} \zeta_{ij\ell kq}^{(\gamma)}\tilde\theta^\prime_{jq}(t_{ij\ell} - t_{ijk}) + \tilde\nu_{iq}\tilde\nu_{jq}^\prime\tilde\theta_{jq}^\prime(T-t_{ij\ell})e^{-\tilde\theta_{iq}\tilde\theta_{jq}^\prime(T-t_{ij\ell})}\}}, \label{est2}
\end{gather}
where $\tau_{ijk}^+=\min\{T,\max\{\tau_{ij}-t_{ik},0\}\}$. Similar equations are available for $\tilde\phi_j^\prime$ and $\tilde\theta_{jq}^\prime$. The full iterative procedure is summarised in Algorithm~\ref{em}.
\endgroup

\begingroup

\begin{algorithm}[!h]
\normalsize
\SetAlgoLined
\SetKw{Until}{until}
\SetKwInput{Input}{Input}
\SetKwInput{Output}{Output}
\Input{initial parameter values $\tilde{\bm\Psi}_0$.}
\Output{model parameters $\tilde{\bm\Psi}$ corresponding to a local maximum of $\lolikt$.}
\For{$k=1,2,\dots$}{
E-step: calculate responsibilities $\xi_\cdot^{(\cdot)}$ and $\zeta_\cdot^{(\cdot)}$ using \eqref{respo} with parameters $\tilde{\bm\Psi}_k$, \\
M-step: calculate $\tilde{\bm\Psi}_{k+1} = \argmax_{\tilde{\bm\Psi}} \mathbb E_{\mvec B, \mvec Z \vert \mathcal H_T, \tilde{\bm\Psi}_k}\{\lolikfull\}$; 
for $r=\infty$, apply \eqref{est1} and \eqref{est2} iteratively, using the \textit{most recent} parameter estimates,
 }
 \Until \emph{convergence in $\lolikt$.}
 \caption{EM 
 algorithm for optimisation of the log-likelihood \eqref{loglik}.}
 \label{em}
\end{algorithm}
\endgroup

\subsection{Inference via gradient ascent methods} \label{sec:adam}

The EM algorithm proposed in the previous section has appealing statistical properties, but it is not scalable for large networks or for large numbers of events, since it requires t$n_{ij}[2 + d + N_i(T) + N_j^\prime(T) + dn_{ij}]$ additional latent variables to be defined for each edge, which is not feasible in most practical applications. 
On the other hand, the log-likelihood in \eqref{loglik} was shown to have a recursive expression for $r=\infty$, which also holds for its gradient with respect to the parameters $\bm\Psi$. Therefore, in order to make the inferential procedure scalable, gradient-based optimisation methods appear to be suitable. 
Gradient ascent methods are usually based on computing the gradient of the log-likelihood function, and iteratively updating the parameter values in the direction of steepest ascent given by the gradient, for a given step size $\eta\in\mathbb R_+$, also known as \textit{learning rate}. One of the main issues of standard gradient ascent for high-dimensional parameter estimation is the choice of the learning rate.
The adaptive moment estimation method \citep[\textit{Adam,}][]{Kingma15} is a popular gradient ascent optimisation algorithm widely used in the machine learning and deep learning communities, which adaptively selects and adjusts learning rates for each parameter. 
Its convergence properties have been extensively studied \citep{Reddi18,Chen19,Zou19}.
In \textit{Adam}, the step sizes are adjusted via exponentially weighted moving averages (EWMA) of the estimated gradient and square gradient (respectively denoted $\vec m$ and $\vec v$, with decay rates $\rho_1,\rho_2\in[0,1]$). Such averages provide estimates for the first and second moment of the gradient respectively; these estimates are then corrected for bias and used to update the parameters in a similar fashion to standard gradient ascent, after adding a small offset $\varepsilon\in\mathbb R_+$ (usually known as smoothing parameter) to the estimate of the second moment, in order to avoid computational issues when its value vanishes towards zero. 
Considering the high-dimensional maximum likelihood estimation of MEG models, Adam appears to be a suitable inferential choice. Alternative gradient ascent techniques for optimisation are surveyed in \cite{Ruder16}. 
In this work, Adam is implemented after adopting a simple re-parametrisation and optimising the logarithm of each parameter, since are all constrained to be positive.
The resulting optimisation procedure is detailed in Algorithm~\ref{adam}. 
The gradient $\vec g = \der{\bm\Psi} \lolik$ of the likelihood \eqref{loglik} 
with respect to $\bm\Psi$ inherits a recursive form from \eqref{meg_lik2}, and therefore it can be calculated in linear time. Explicit forms 
for $d=1$ and $r=\infty$ are derived in the Appendix.

\begingroup

\begin{algorithm}[!h]
\normalsize
\SetAlgoLined
\SetKw{Until}{until}
\SetKwInput{Input}{Input}
\SetKwInput{Output}{Output}
\Input{step size $\eta\in\mathbb R_+$, decay rates $\rho_1,\rho_2\in(0,1)$, smoothing parameter $\varepsilon\in\mathbb R_+$, initial parameter values $\bm\Psi_0$.}
\Output{model parameters $\bm\Psi$ corresponding to a local maximum of $\lolik$.}
Initialise estimates of the first and second moment of the gradient: 
$\vec m_0=\vec 0, \vec v_0=\vec 0,$ \\
\For{$k=1,2,\dots$}{
calculate gradient $\vec g_k = \left.\der{\bm\Psi} \lolik\right|_{\bm\Psi=\bm\Psi_{k-1}}$, evaluated at $\bm\Psi_{k-1}$, \\ 
update EWMA estimate of first moment: $\vec m_k = \rho_1\vec m_{k-1} + (1-\rho_1)(\vec g_k\times\bm\Psi_{k-1})$, \\
update second moment: $\vec v_k = \rho_2\vec v_{k-1} + (1-\rho_2)[(\vec g_k\times\bm\Psi_{k-1})\times(\vec g_k\times\bm\Psi_{k-1})]$, \\ 
update parameters: $\bm\Psi_k = \bm\Psi_{k-1}\times\exp\left\{\eta \vec m_t \big/ (1-\rho_1^k)\left(\sqrt{\vec v_t/ (1-\rho_2^k)} + \varepsilon\right)\right\}$, 
 }
 \Until \emph{convergence in $\lolik$}. \\
\nonl\emph{Sums, products, quotients, exponentials, and square roots are applied element-wise.} 
 \caption{\textit{Adam} 
 algorithm for optimisation of the log-likelihood \eqref{loglik}.}
 \label{adam}
\end{algorithm}
\endgroup

\section{Simulation and assessment of the goodness-of-fit} \label{simulation}

In order to validate the inferential procedure, it is necessary to simulate data from the MEG model \eqref{megg}, which can be interpreted as an extended 
multivariate Hawkes process where some of the parameters are shared across the individual processes. 
Therefore, 
simulating MEG models is possible under the framework described in \cite{Ogata81}, and follows the standard technique 
of simulation via \textit{thinning}. 
The procedure is described in Algorithm~\ref{algo_sim}. 

\begingroup

\begin{algorithm}[!h]
\normalsize
\SetAlgoLined
set $t^\star=0$,\\
 \Repeat{$t^\star>T$}{
  set $\lambda^\star = \sum_{i=1}^n\sum_{j=1}^n
  \lambda_{ij}(t^\star_+)$, where $t^\star_+$ denotes the limit from the right, \\
  generate the inter-arrival time $q=-\log(u)/\lambda^\star$, where $u\sim\mathrm{Uniform}[0,1]$, \\ 
  obtain the candidate arrival time $t^\star\leftarrow t^\star+q$, \\
  assign the arrival time $t^\star$ to the edge $(i,j)$ with probability $
  \lambda_{ij}(t^\star_-)/\lambda^\star$, and do not assign to any edge with probability $1-\sum_{i=1}^n\sum_{j=1}^n
  \lambda_{ij}(t^\star_-)/\lambda^\star$, where $t^\star_-$ denotes the limit from the left. \\
 }
 \caption{Simulation of a MEG in $[0,T]$.}
 \label{algo_sim}
\end{algorithm}
\endgroup

Furthermore, it is possible to assess the performance of the inferential procedure by evaluating the goodness-of-fit from out-of-sample events. If the model parameters are estimated only from the event times obtained in $[0,T^\star]$, with $T^\star<T$, using Algorithm~\ref{adam}, the goodness-of-fit can then be evaluated from the event times in $(T^\star,T]$. Goodness-of-fit measures can be calculated from functions of the compensator function for the model.
Given the conditional intensity $\lambda_{ij}(t)$, the compensator $\Lambda_{ij}(t)$  is:
\begin{equation}
\Lambda_{ij}(t) = \int_{\tau_{ij}}^t \lambda_{ij}(s)\mathrm ds.
\end{equation}
Examples of compensator functions for some MEG models, for $t=T$, can be found in the Appendix. 
Given arrival times $t_1,\dots,t_{n_{ij}}$ on the edge $(i,j)$, under the null hypothesis of correct specification of the conditional intensity $\lambda_{ij}(t)$, by time rescaling theorem \citep[see, for example,][]{Brown02} $\Lambda_{ij}(t_1),\dots, \Lambda_{ij}(t_{n_{ij}})$ are event times of a homogeneous Poisson process with unit rate. It follows that 
the upper tail $p$-values
\begin{equation}
p_{ijk} = \exp\{-\Lambda_{ij}(t_k)+\Lambda_{ij}(t_{k-1})\} = \exp\left\{-\int_{t_{k-1}}^{t_k} \lambda_{ij}(s) \dif s\right\} \label{pval}
\end{equation}
follow a standard uniform distribution under the null hypothesis. Therefore, given the estimates of the conditional intensity functions obtained from the arrival times in $[0,T^\star]$, approximately uniform $p$-values for the test event times in $(T^\star,T]$ should be observed if the model is specified and estimated correctly.

\section{Applications and results} \label{results}

In this section, the MEG model is tested on simulated network data and on two real world computer network datasets: the Enron e-mail network, and a bipartite graph obtained from network flow data collected at Imperial College London. Across the experiments, the decay rates $(\rho_1,\rho_2)$ in Algorithm~\ref{adam} have been set to $(0.9,0.99)$, and $\varepsilon=10^{-8}$.

\subsection{Simulated events on small fully connected graphs} \label{sim_full}

\begin{figure}[!t]
\centering
\begin{subfigure}[t]{.495\textwidth}
\centering
\caption{Baseline}
\includegraphics[width=0.975\textwidth]{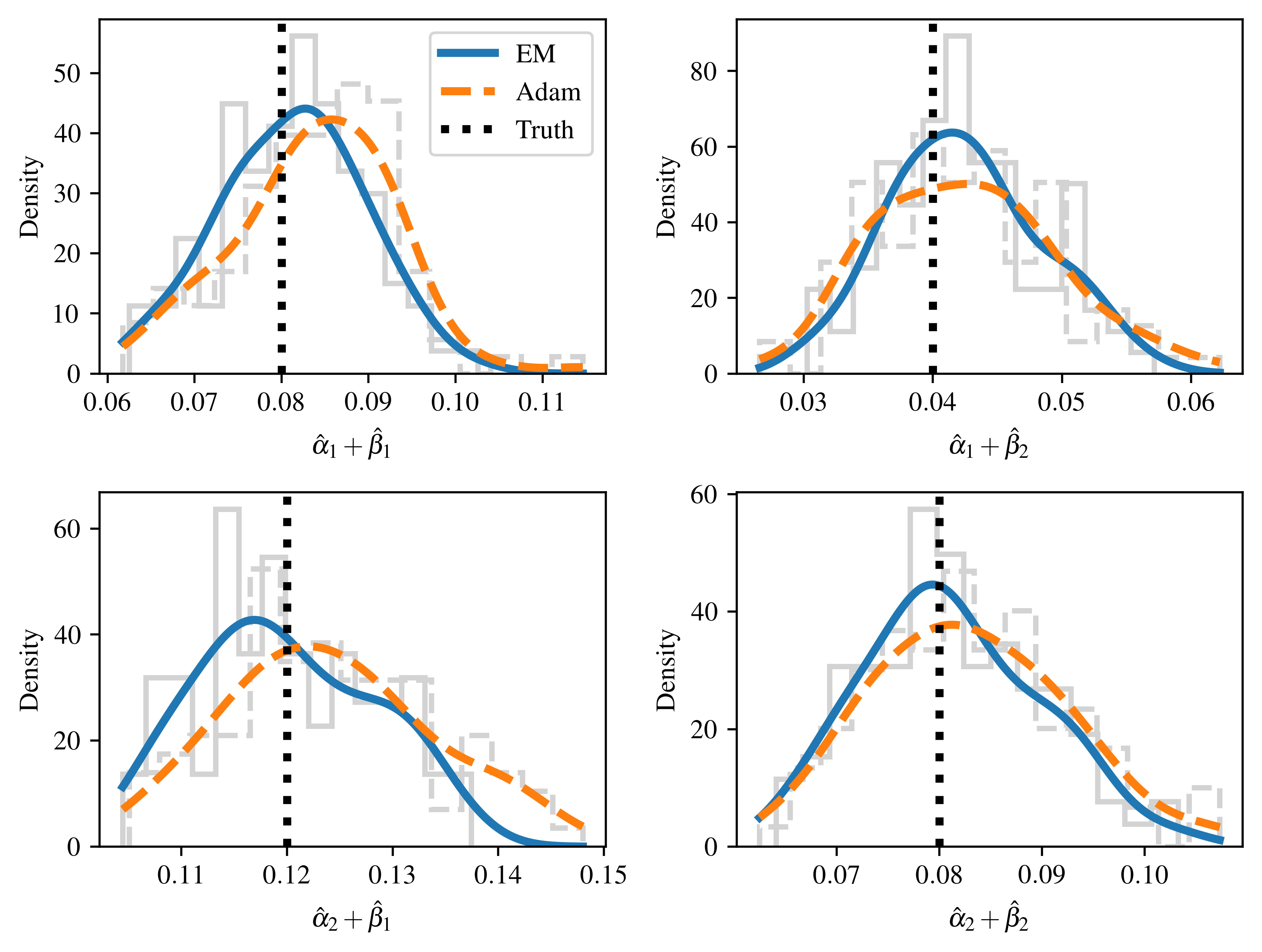}
\label{main_alpha_beta}
\end{subfigure}
\begin{subfigure}[t]{0.495\textwidth}
\centering
\caption{Jump}
\includegraphics[width=0.975\textwidth]{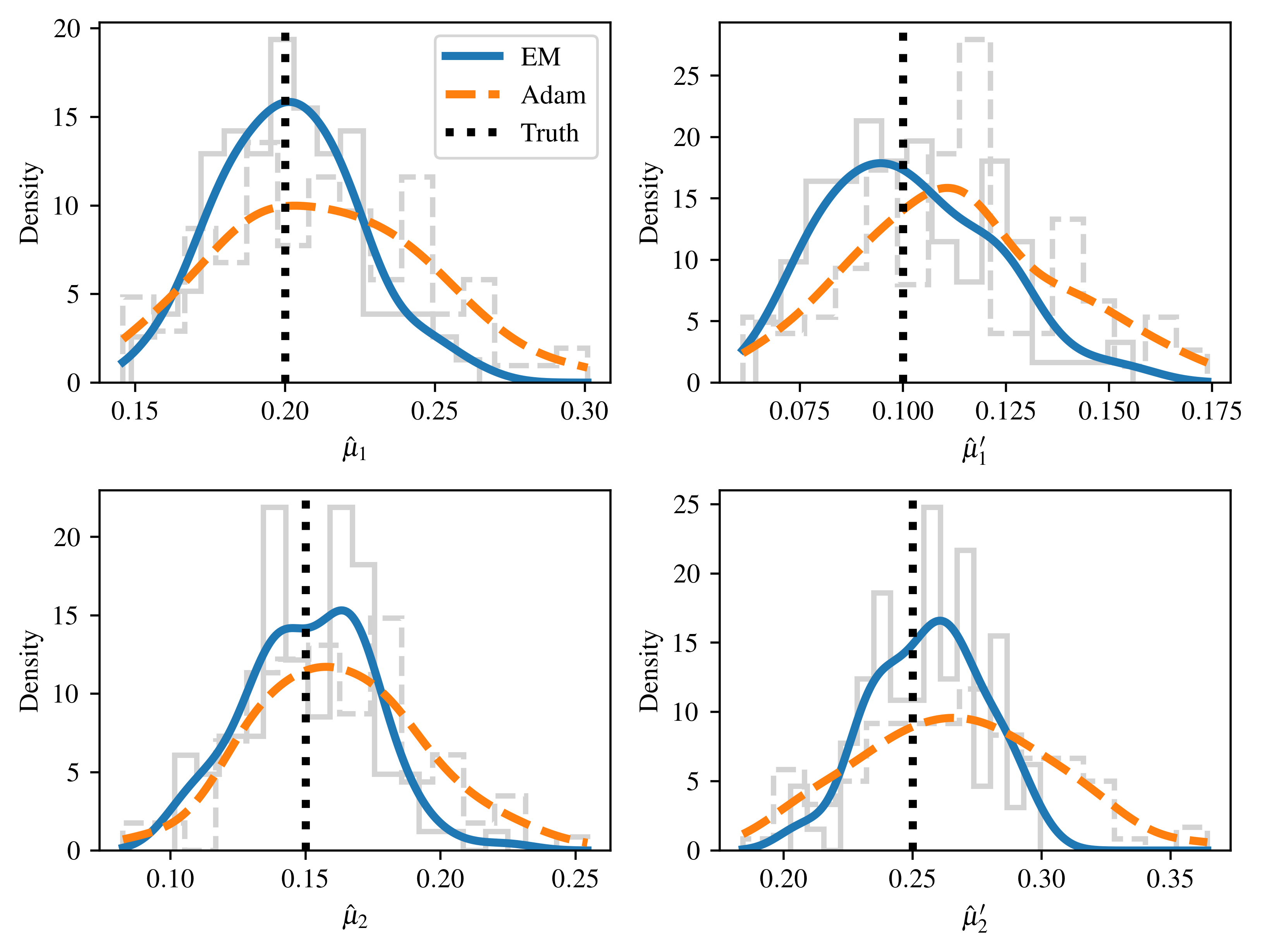}
\label{main_jump}
\end{subfigure}
\begin{subfigure}[t]{0.495\textwidth}
\centering
\caption{Decay}
\includegraphics[width=0.975\textwidth]{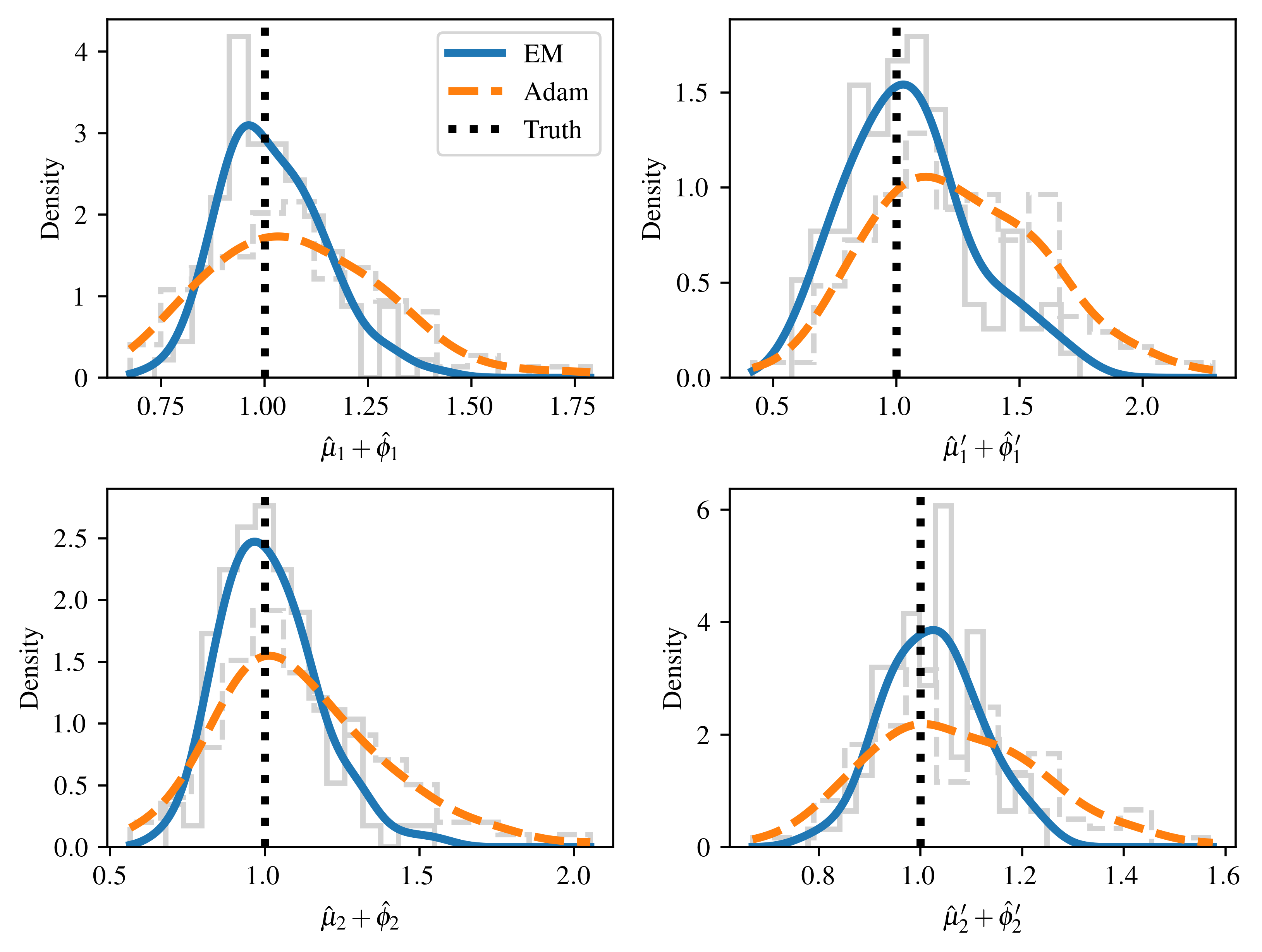}
\label{main_decay}
\end{subfigure}
\begin{subfigure}[t]{0.495\textwidth}
\centering
\caption{KS scores}
\includegraphics[width=0.975\textwidth]{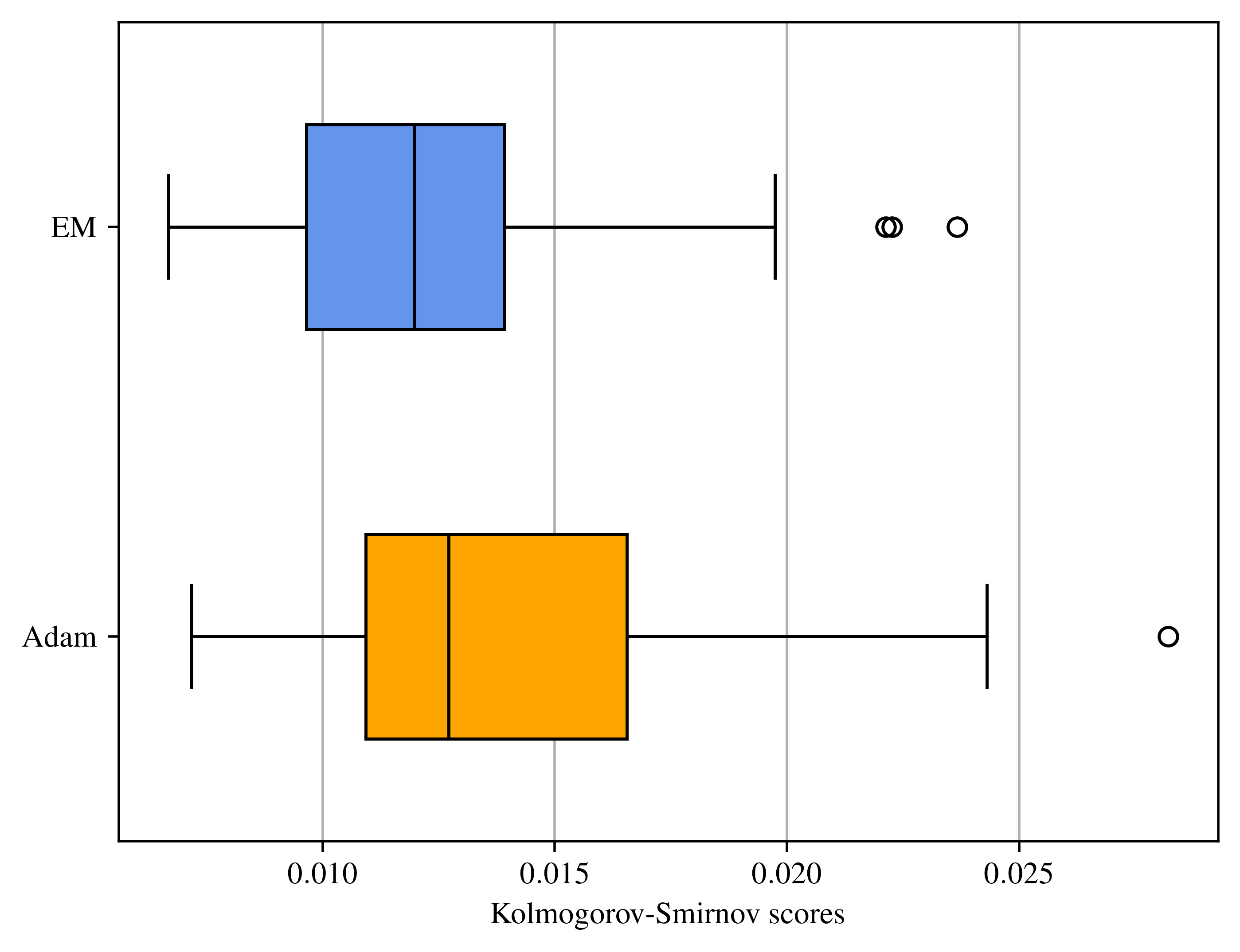}
\label{main_ks}
\end{subfigure}
\caption{Histograms (with corresponding kernel density estimates) of parameter estimates and boxplots of KS scores obtained using EM and Adam from $100$ simulations of $\numprint{3000}$ events on a fully connected MEG with $n=2$, $\lambda_{ij}(t)=\alpha_i(t) + \beta_j(t)$, $r=\infty$, $\bm\alpha=[0.01,0.05]$, $\bm\beta=[0.07,0.03]$, $\bm\mu=[0.2,0.15]$, $\bm\mu^\prime=[0.1,0.25]$, $\bm\phi=[0.8,0.85]$, $\bm\phi^\prime=[0.9,0.75]$.}
\label{main_sim}
\end{figure}

In order to evaluate Algorithm~\ref{em} and \ref{adam} and their performance at estimating MEG models, simulated network data are initially used. In this section, a small fully connected directed graph with $n=2$ is considered. Two types of mutually exciting graphs with $r=\infty$ and $\tau_{ij}=0$ are generated: 
 \begin{enumerate*}[label=(\roman*)] \item MEGs with $\lambda_{ij}(t)=\alpha_i(t) + \beta_j(t)$, $\alpha_i(t)$ and $\beta_j(t)$ as in \eqref{node_process} and \eqref{scaled_expo}, with $\bm\alpha=[0.01,0.05]$, $\bm\beta=[0.07,0.03]$, $\bm\mu=[0.2,0.15]$, $\bm\mu^\prime=[0.1,0.25]$, $\bm\phi=[0.8,0.85]$, $\bm\phi^\prime=[0.9,0.75]$, and
 \item MEGs with $\lambda_{ij}(t)=\gamma_{ij}(t)$, \textit{cf.} \eqref{edge_process}, $\bm\gamma=[0.1,0.5]$, $\bm\gamma^\prime=[0.1,0.3]$, $\bm\nu=[0.6,0.4]$, $\bm\nu^\prime=[0.5,0.25]$, $\bm\theta=[0.4,0.6]$, $\bm\theta^\prime=[0.5,0.75]$. 
 \end{enumerate*}
 For each of the two MEG models, $\numprint{3000}$ events are simulated using Algorithm~\ref{algo_sim}, and the process parameters are estimated from the simulated events via EM (Algorithm~\ref{em}) and Adam (Algorithm~\ref{adam}, with $\eta=0.05$), initialising the parameters at random from a uniform distribution in $(0.1,1)$. The estimation procedure is repeated $5$ times from different random initialisation points, and the final estimates corresponding to the highest log-likelihood are retained. Using the estimated parameters, the $p$-values \eqref{pval} are then calculated for all simulated events. Finally, the Kolmogorov-Smirnov (KS) score against the uniform distribution is calculated on those $p$-values. The procedure is then repeated $100$ times, obtaining a set of parameter estimates for each simulated MEG. 
 
The results are plotted in Figures~\ref{main_sim} and~\ref{inter_sim}, which report the histograms and kernel density estimates of identifiable transformations of the parameters for each of the four network edges, obtained using the EM and Adam algorithms, compared to the true value of the parameters. 
Furthermore, the figures report the boxplots of the KS scores. 
 
\begin{figure}[!t]
\centering
\begin{subfigure}[t]{.495\textwidth}
\centering
\caption{Baseline}
\includegraphics[width=0.975\textwidth]{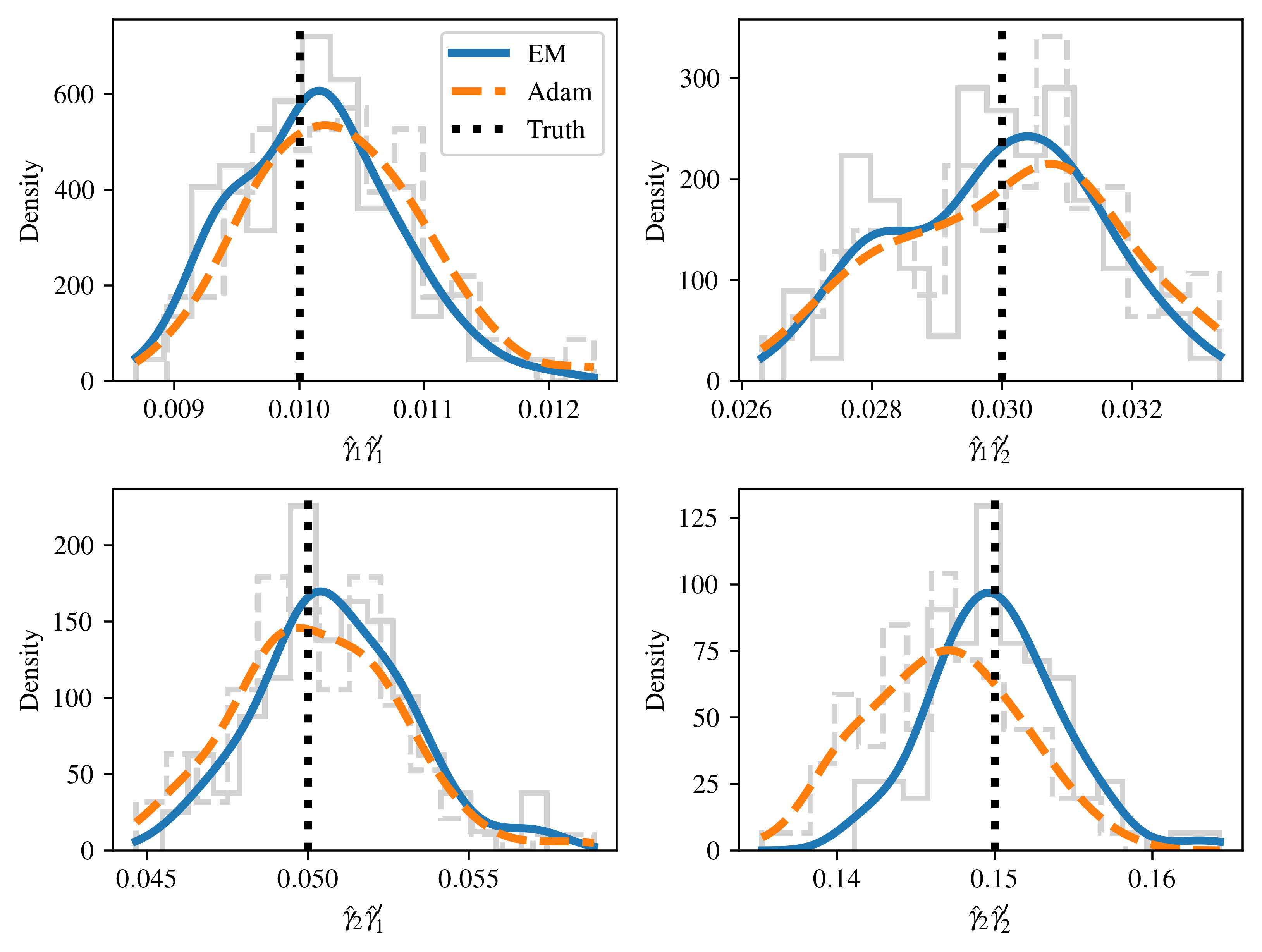}
\label{inter_gamma}
\end{subfigure}
\begin{subfigure}[t]{0.495\textwidth}
\centering
\caption{Jump}
\includegraphics[width=0.975\textwidth]{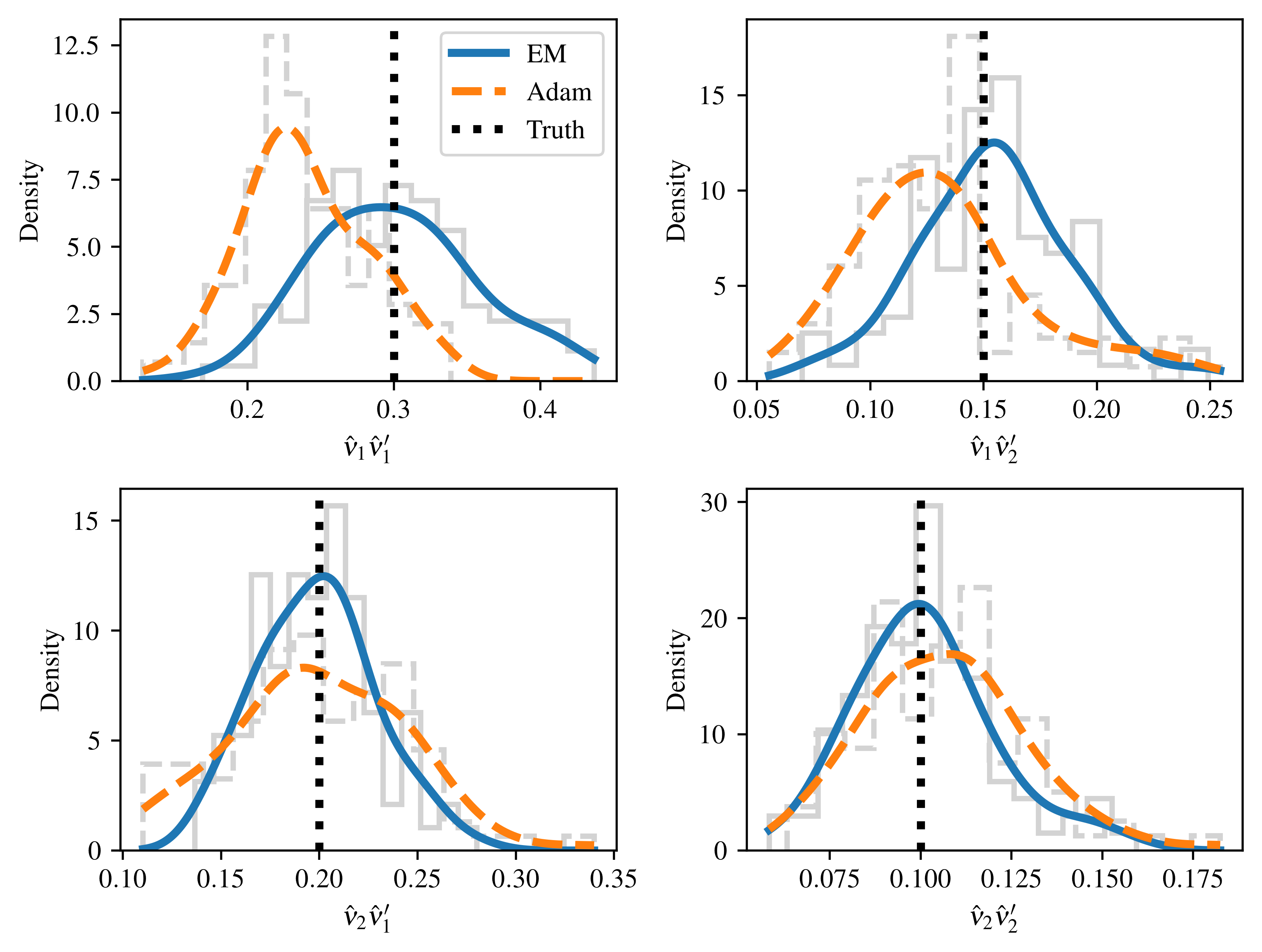}
\label{inter_jump}
\end{subfigure}
\begin{subfigure}[t]{0.495\textwidth}
\centering
\caption{Decay}
\includegraphics[width=0.975\textwidth]{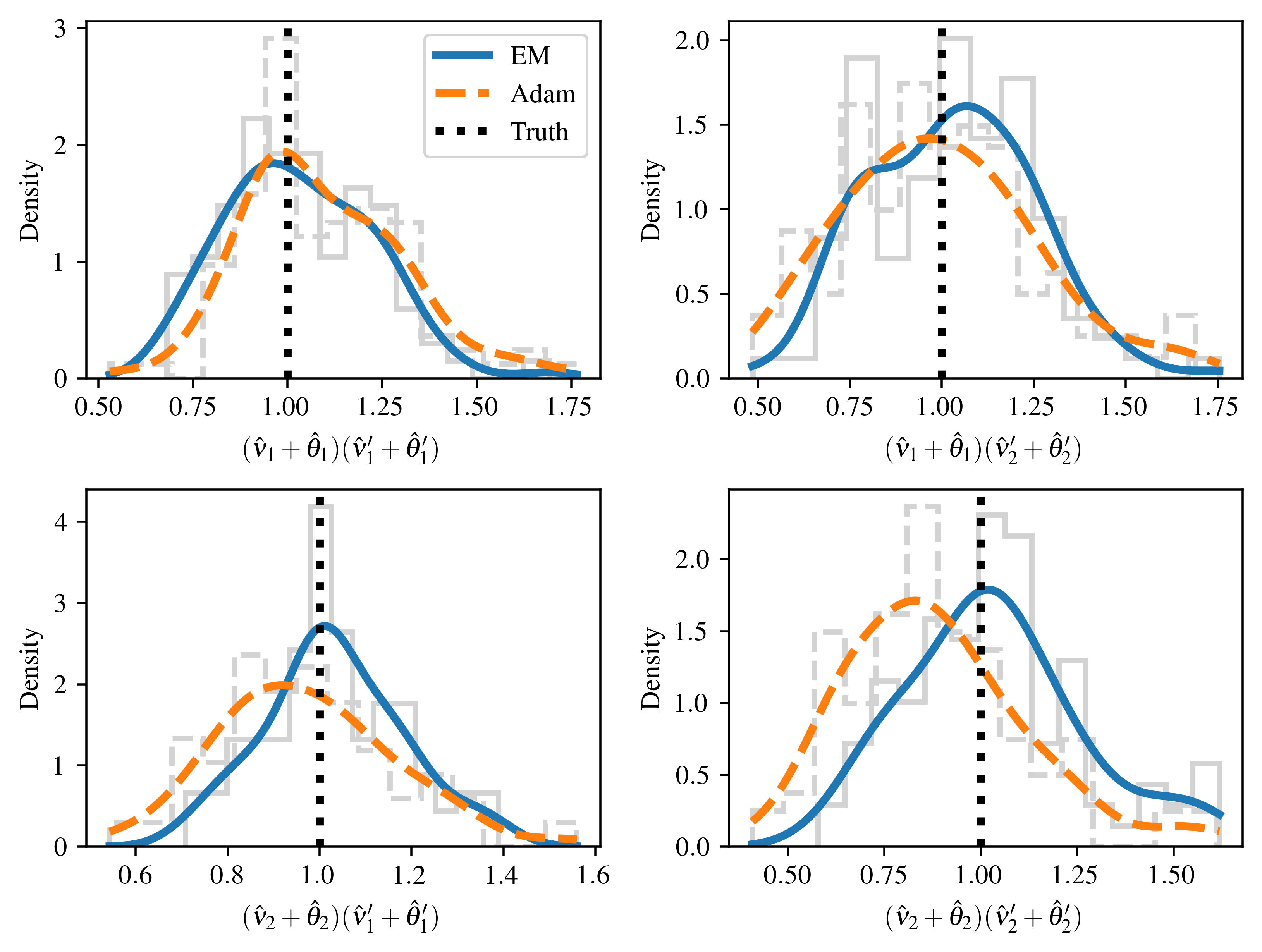}
\label{inter_decay}
\end{subfigure}
\begin{subfigure}[t]{0.495\textwidth}
\centering
\caption{KS scores}
\includegraphics[width=0.975\textwidth]{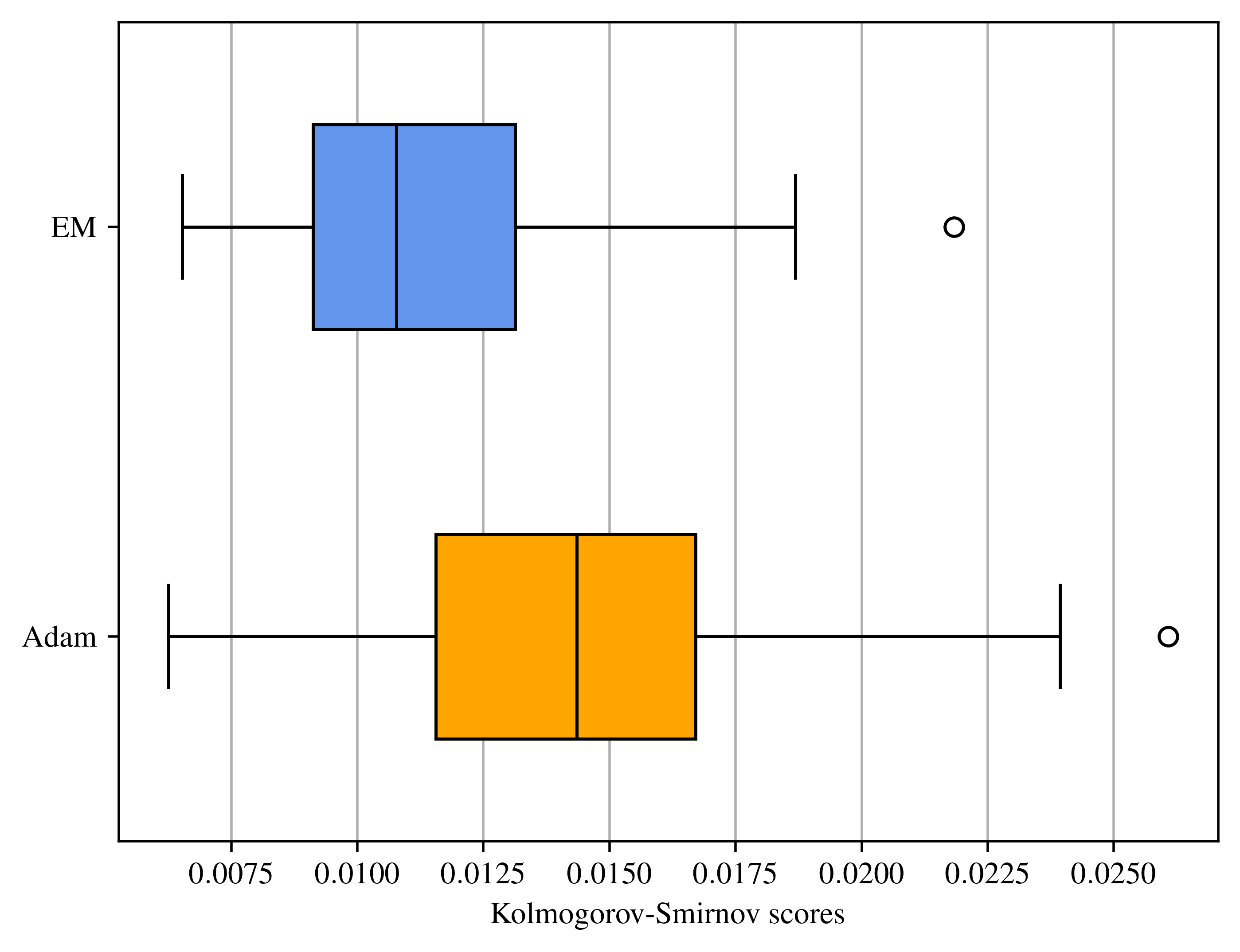}
\label{inter_ks}
\end{subfigure}
\caption{Histograms (with corresponding kernel density estimates) of parameter estimates and boxplots of KS scores obtained using EM and Adam from $100$ simulations of $\numprint{3000}$ events on a fully connected MEG with $n=2$, $\lambda_{ij}(t)=\gamma_{ij}(t)$, $r=\infty$, $\bm\gamma=[0.1,0.5]$, $\bm\gamma^\prime=[0.1,0.3]$, $\bm\nu=[0.6,0.4]$, $\bm\nu^\prime=[0.5,0.25]$, $\bm\theta=[0.4,0.6]$, $\bm\theta^\prime=[0.5,0.75]$.}
\label{inter_sim}
\end{figure}

Overall, it appears that the results obtained using Adam are only marginally worse than those obtained using the EM algorithm. In particular, the distributions of estimates obtained using the two methodologies appear to be roughly centred around the true value of the parameters, but the EM estimates appear to be slightly more precise and accurate when compared to Adam. In both cases, the KS scores are extremely small, demonstrating an excellent fit.
Furthermore, it is possible to compare the performance of the two inferential algorithms when the number of events increases: Figure~\ref{main_asy} reports the histograms of estimates of the decay $\bm\mu+\bm\phi$ obtained using only $250, 500, \numprint{1000}$ and $\numprint{2000}$ of the $\numprint{3000}$ simulated events on each graph. The performance of both estimation procedures improves in terms of KS scores and variance of estimates when more observations are available. 

\begin{figure}[!t]
\centering
\begin{subfigure}[t]{.24\textwidth}
\centering
\caption{$m=250$}
\includegraphics[width=0.975\textwidth]{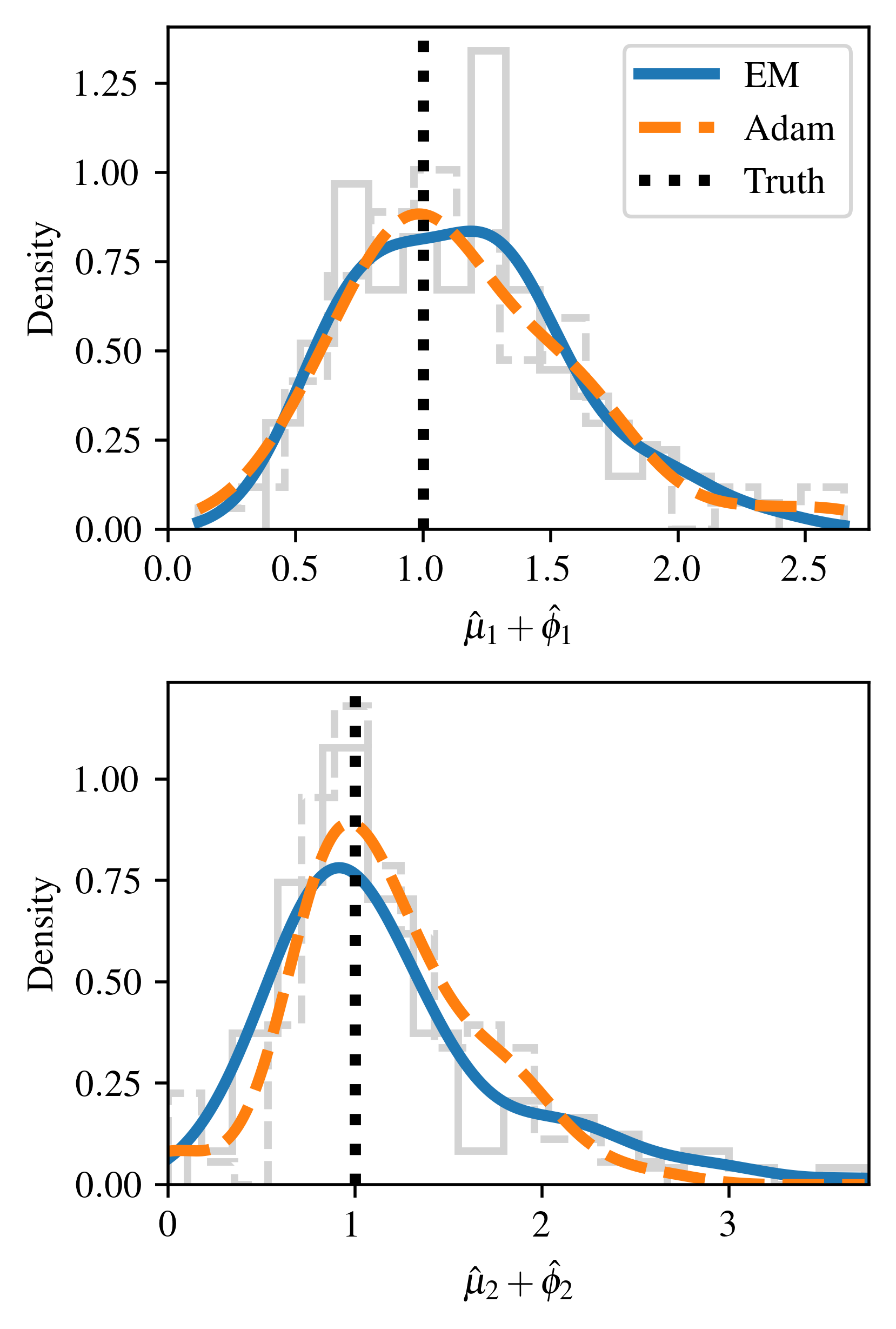}
\includegraphics[width=0.975\textwidth]{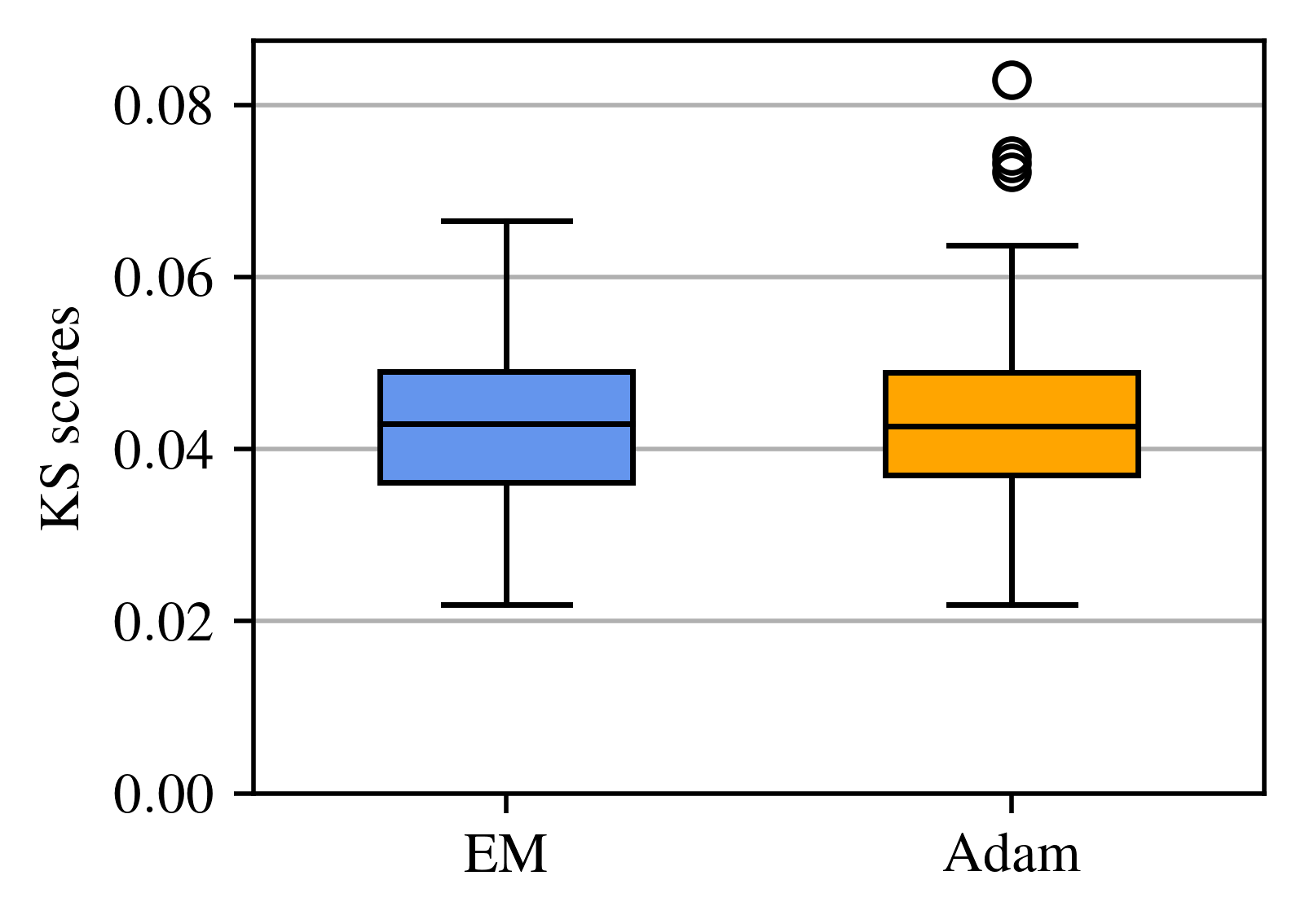}
\end{subfigure}
\begin{subfigure}[t]{.24\textwidth}
\centering
\caption{$m=500$}
\includegraphics[width=0.975\textwidth]{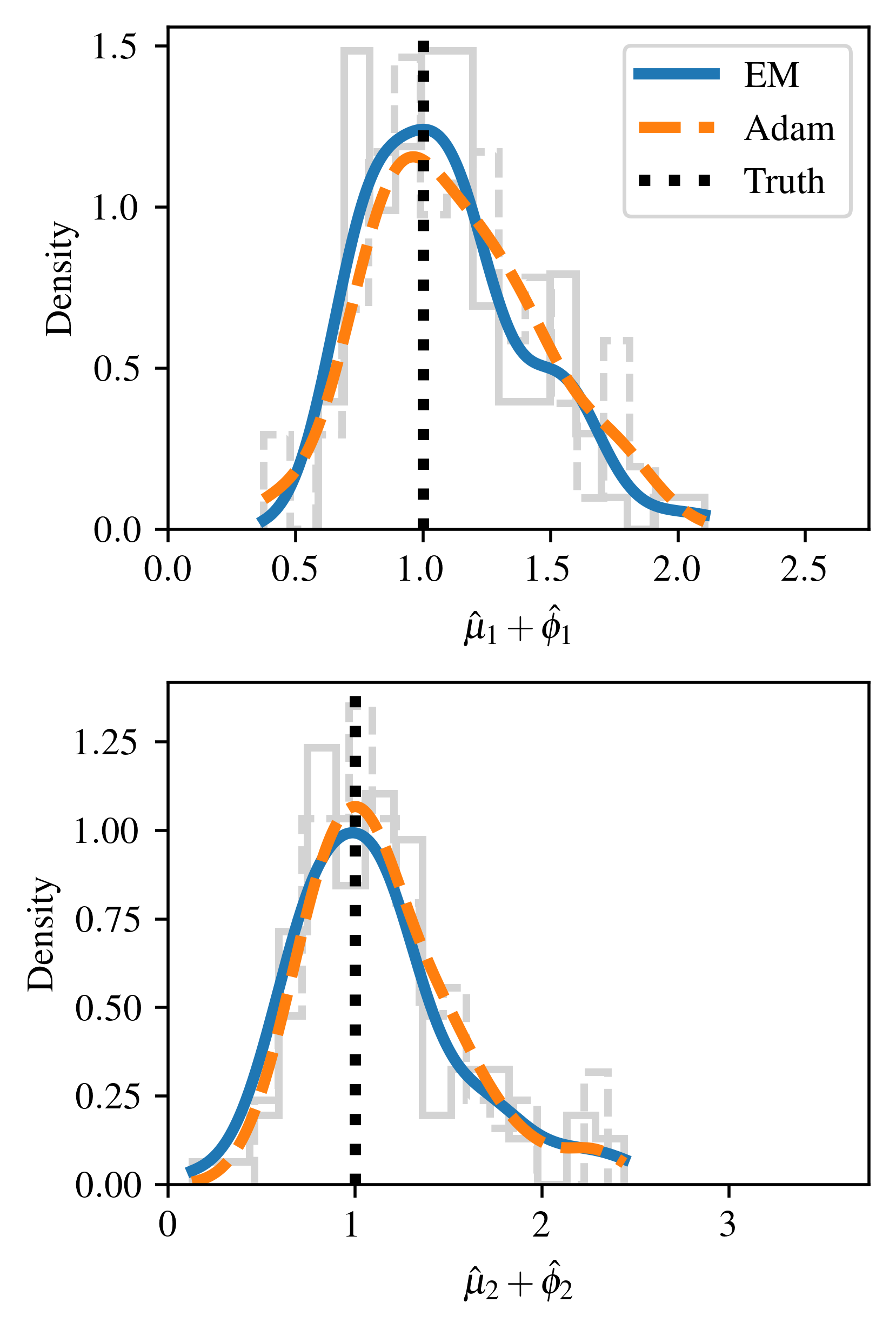}
\includegraphics[width=0.975\textwidth]{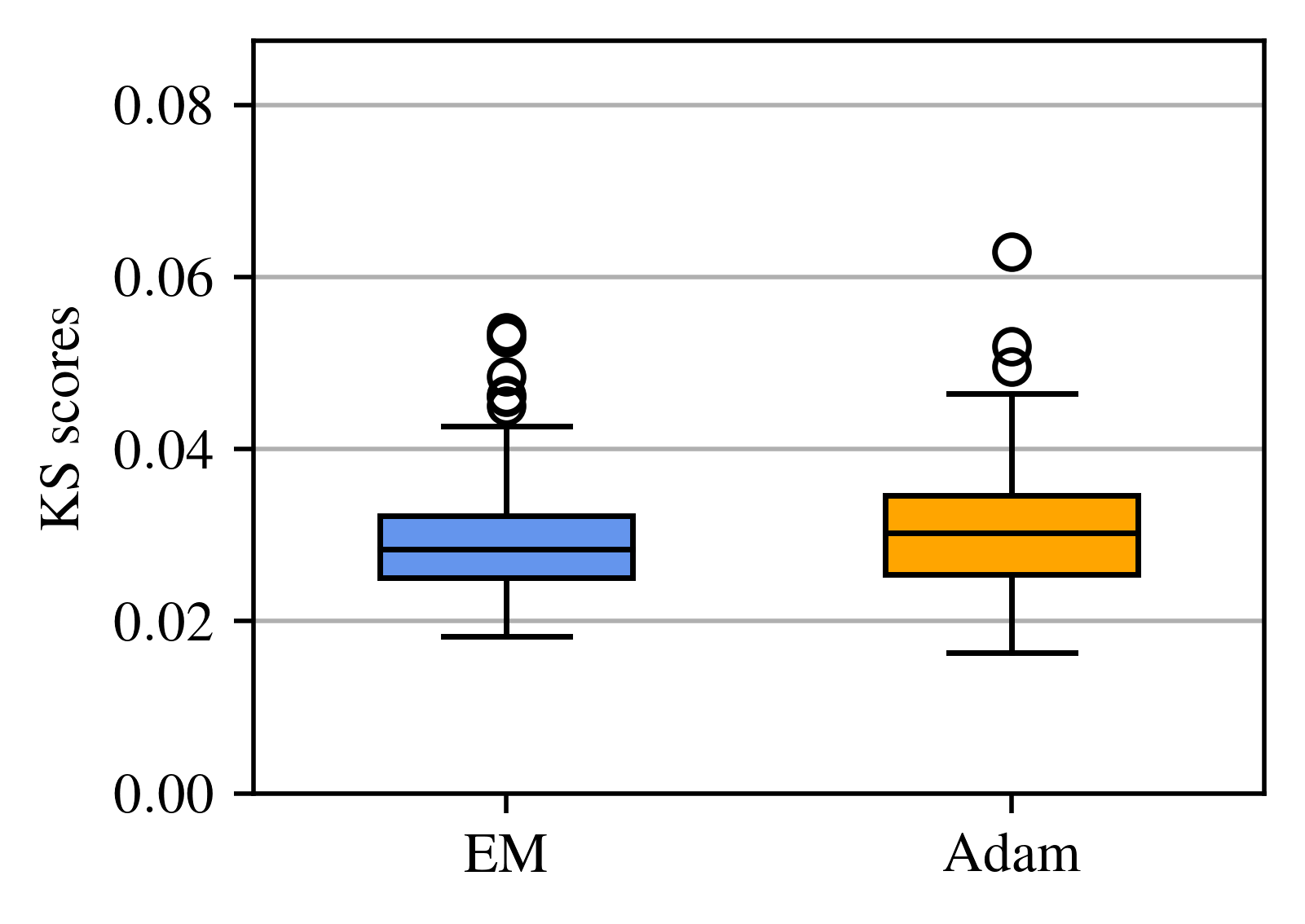}
\end{subfigure}
\begin{subfigure}[t]{.24\textwidth}
\centering
\caption{$m=\numprint{1000}$}
\includegraphics[width=0.975\textwidth]{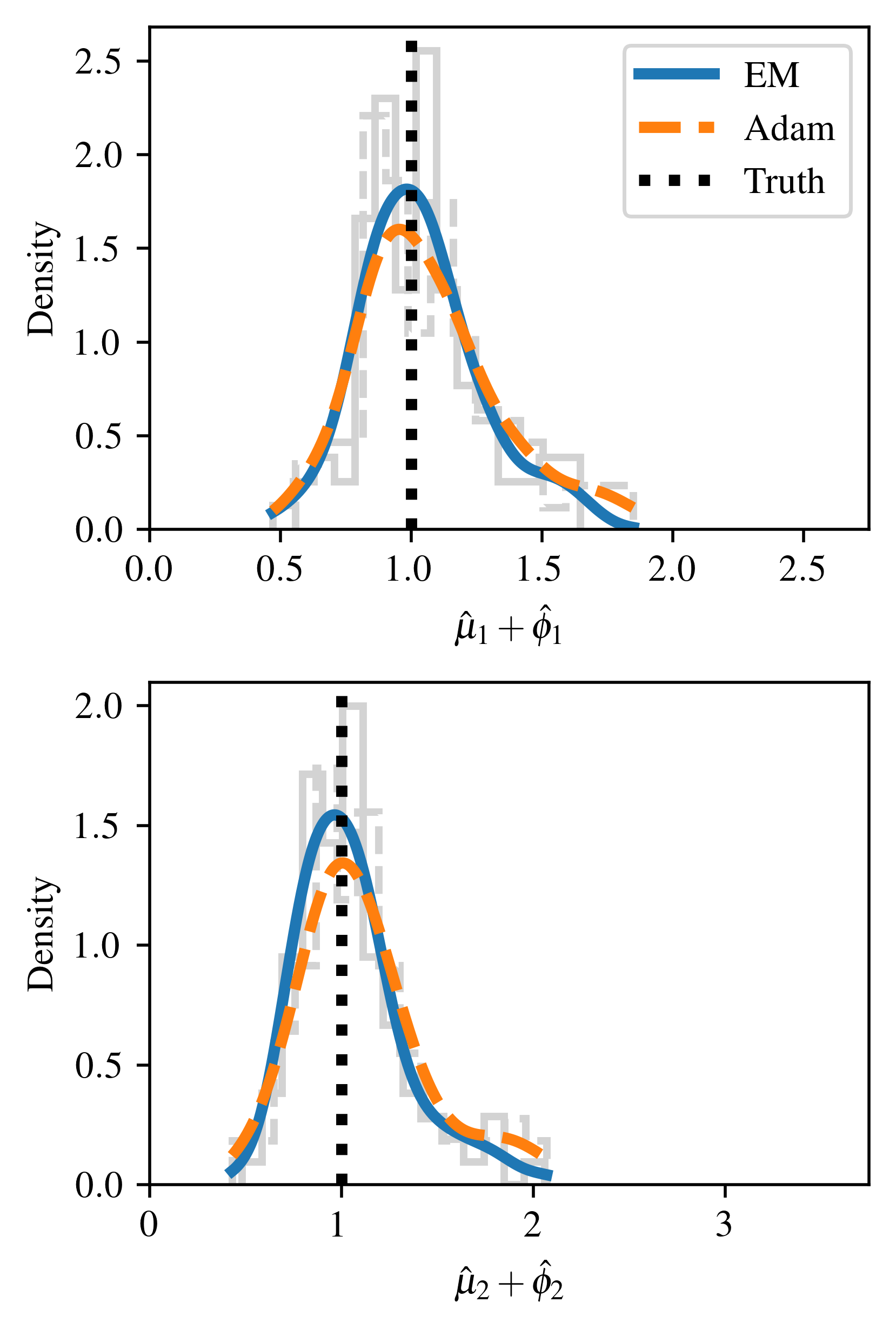}
\includegraphics[width=0.975\textwidth]{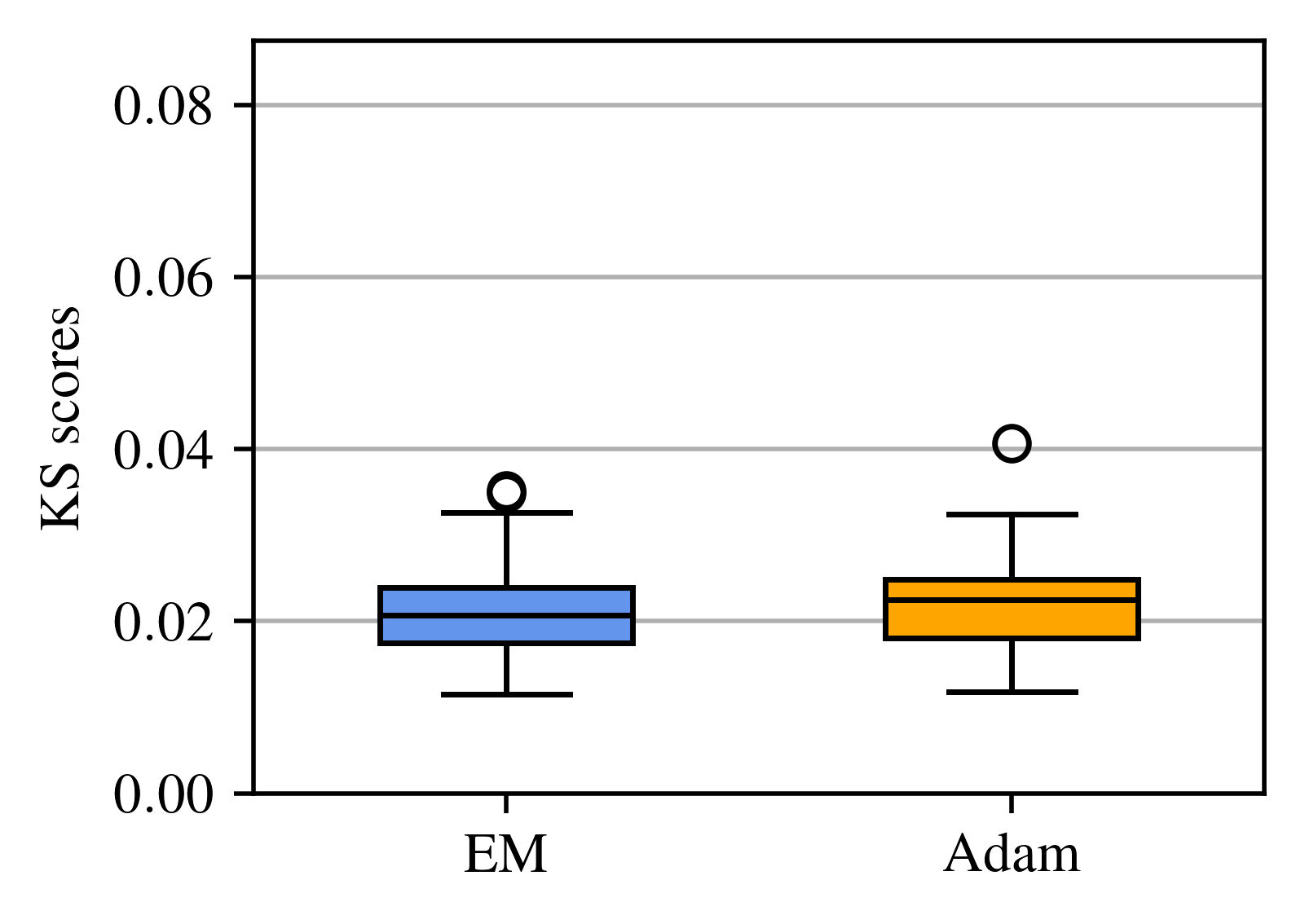}
\end{subfigure}
\begin{subfigure}[t]{.24\textwidth}
\centering
\caption{$m=\numprint{2000}$}
\includegraphics[width=0.975\textwidth]{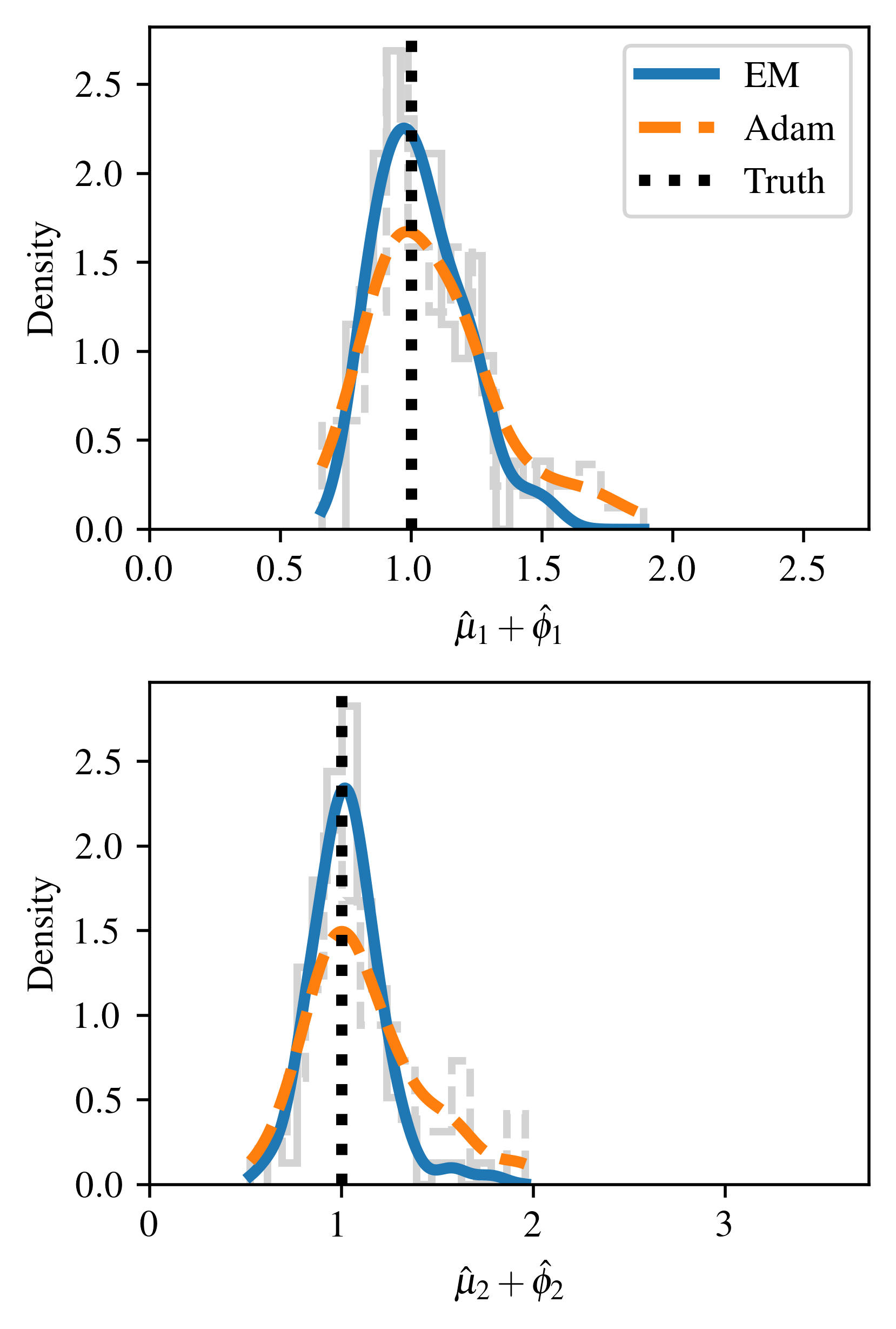}
\includegraphics[width=0.975\textwidth]{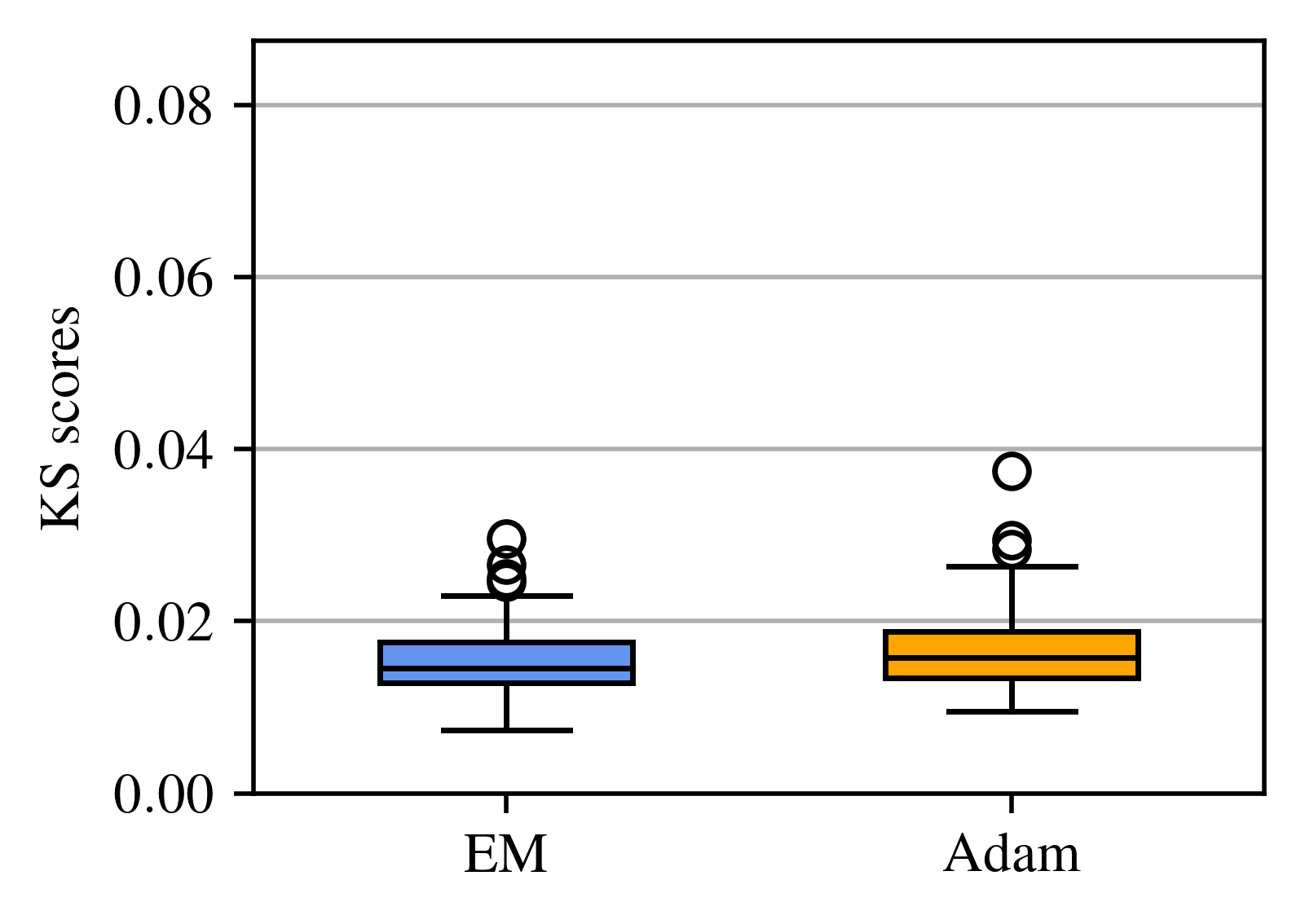}
\end{subfigure}
\caption{Histograms (with corresponding kernel density estimates) of estimates for $\bm\mu+\bm\phi$ and boxplots of KS scores obtained using EM and Adam from $100$ simulations from the same model as Figure~\ref{main_sim}, with $m\in\{250,500,\numprint{1000},\numprint{2000}\}$ events.}
\label{main_asy}
\end{figure}

The main advantage of Adam over the EM algorithm for $r=\infty$ is mainly given by the recursive form of the log-likelihood, which enables calculations in linear time on each edge for the gradients. On the other hand, the EM algorithm requires a quadratic number of additional latent variables defined on each edge, which becomes unsustainable for efficient inference on large graphs. Therefore, for the remainder of this work, Adam will be used, primarily because of computational reasons. 

\subsection{Simulated events on Erd\H{o}s-R\'{e}nyi graphs} \label{sim_data}

In order to further evaluate estimation of MEG models, events on an Erd\H{o}s-R\'{e}nyi graph are also simulated. 
First, an adjacency matrix is simulated from an Erd\H{o}s-R\'{e}nyi graph with $n=10$ nodes, such that $A_{ij}\sim\mathrm{Bernoulli}(p)$, with $p=1/4$. For the edges such that $A_{ij}=1$, then $\tau_{ij}=0$, otherwise if $A_{ij}=0$, then $\tau_{ij}=\infty$. 
Second, $m=\numprint{2500}$ event times are generated from a MEG model with $r=\infty$ using Algorithm~\ref{algo_sim}, with parameters in $\bm\Psi$ sampled at random from uniform distributions, restricted to the following ranges: $\alpha_i,\beta_j\in(10^{-5},10^{-4}),\ \mu_i,\mu_j^\prime,\phi_i,\phi^\prime_j\in(10^{-2},10^{-1}),\ \gamma_{i\ell},\gamma_{j\ell}^\prime\in(10^{-5},10^{-1}),\ \nu_{i\ell},\nu_{j\ell}^\prime\in(10^{-2},1)$, and $\theta_{i\ell}=1-\nu_{i\ell}$, $\theta_{j\ell}^\prime=1-\nu_{j\ell}^\prime$. In the simulation, the expected number of events per active edge is $m/[pn(n-1)]\approx 111$. Algorithm~\ref{adam} is used to estimate $2n\times 6=120$ parameters, with learning rate $\eta=0.1$, after a random initialisation from the same uniform distributions used in the data generating process. The entire procedure is repeated $100$ times. 

A second simulation is conducted for an Erd\H{o}s-R\'{e}nyi graph graph with $n=20$ nodes and $p=1/4$, simulating $m=\numprint{10000}$ events from a MEG model with interaction term only, corresponding to $\lambda_{ij}(t)=A_{ij}\gamma_{ij}(t)$, with $r=1$ and $d=5$. A minor modification is made to the range of the uniform distributions for sampling some of the interaction term parameters: $\gamma_{i\ell},\gamma_{j\ell}^\prime\in(10^{-5},10^{-1}),\ \nu_{i\ell},\nu_{j\ell}^\prime,\theta_{i\ell},\theta_{j\ell}^\prime\in(10^{-2},1)$. Despite the simpler form of the intensity functions $\lambda_{ij}(t)$, 
more parameters must be estimated ($2n\times 3d=600$) compared to the first simulation, and the expected number of connections per edge is only $105$. 
As before, $100$ MEGs are generated, and Adam (Algorithm~\ref{adam}) is used to estimate the parameters, with learning rate $\eta=10^{-3}$. 
The resulting boxplots of the KS test obtained for the two simulations are plotted in Figure~\ref{ks_sim}. 
Both boxplots demonstrate that the algorithm is able to recover sensible estimates of the parameter values, resulting in small KS scores, corresponding to a good model fit.

\begin{figure}[!t]
\centering
\includegraphics[width=.8\textwidth]{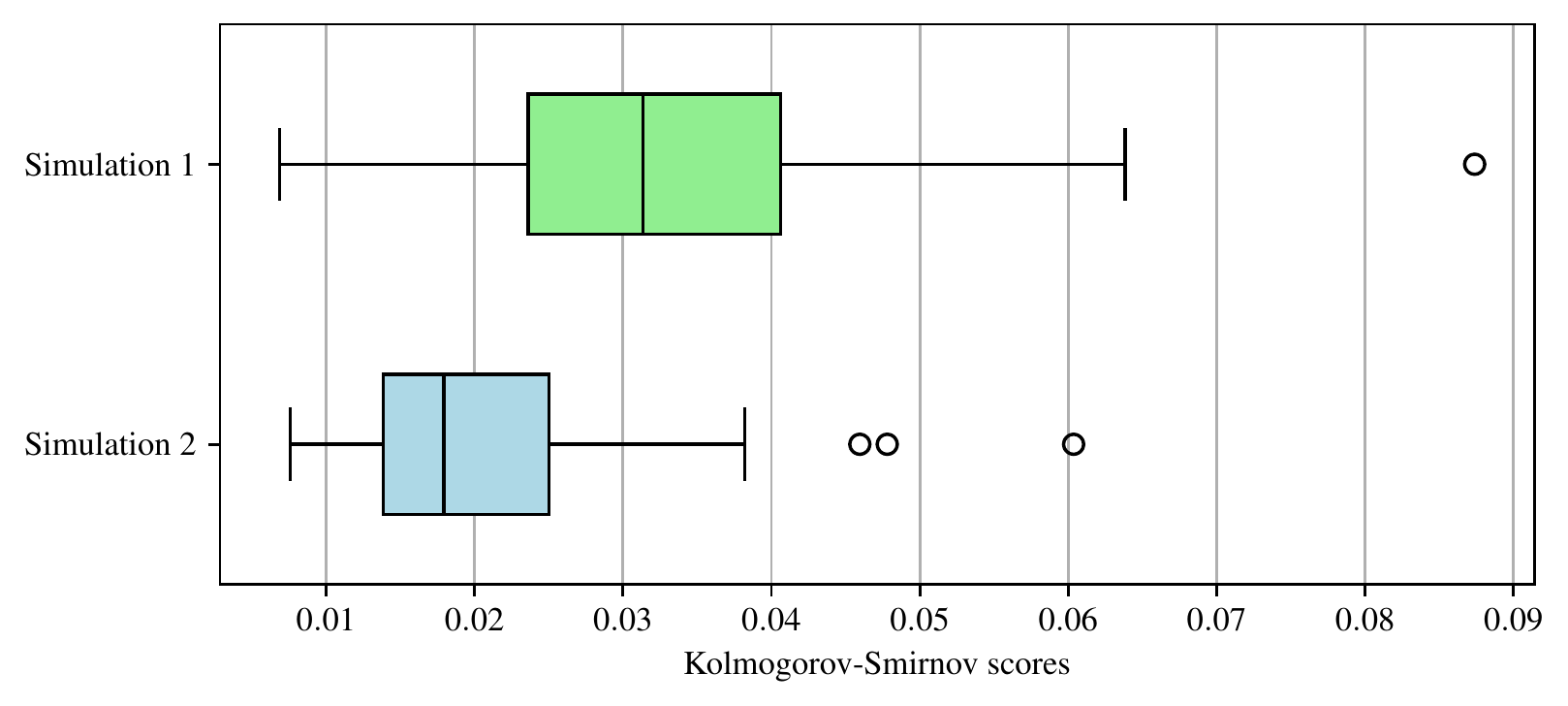}
\caption{Boxplots of the Kolmogorov-Smirnov scores obtained for the two simulations described in Section~\ref{sim_data}.}
\label{ks_sim}
\end{figure}

\subsection{Enron e-mail network} \label{enron_data}

The Enron e-mail network collection
is a record of e-mails exchanged between the employees of Enron Corporation before its bankruptcy. These data have already been demonstrated to be well-modelled as self-exciting point processes by \cite{Fox16}. In this article, the version of these data\footnote{The data are freely available at \url{http://www.cis.jhu.edu/~parky/Enron/}.} used in \cite{Priebe05} is analysed, where e-mails recorded multiple times have been used only once, and e-mails with incorrectly recorded sent times (coded in the data with 9pm, 31 December 1979) have been removed. After such pre-processing, the e-mail data consist of $\numprint{34427}$ distinct triplets $(x_k,y_k,t_k)$, corresponding to messages exchanged between $n=184$ employees between November 1998 and June 2002, forming a total of $\numprint{3007}$ graph edges. Note that some of the emails are sent to multiple receivers, and only $\numprint{18031}$ unique event times are observed, implying that on average each e-mail is sent to approximately $1.90$ nodes. 

Because an e-mail can have multiple recipients, and because the event times are recorded to the nearest second, the likelihood \eqref{loglik} must be adapted slightly to handle tied arrival times. 
An approach used by \citet[Section 8.1]{PriceWilliams19} is followed, with the arrivals modelled by an analogous discrete time process: 
In particular, arrivals at time $t$ are assumed to contribute to the 
the intensities $\lambda_{ij}(\cdot)$ from time $t+\dif t$ onwards, where $\dif t$ is the sampling interval, equal to one second in this example. 
The $p$-values of the process are approximated using \eqref{pval}, following \cite{Fox16}.

The model is trained on $\numprint{30704}$ e-mails sent before 1st December, 2001, and tested on the remaining $\numprint{3723}$ e-mails. In the training set, $\numprint{2720}$ edges are observed, and $\numprint{811}$ in the test set, of which $\numprint{287}$ are \textit{not} observed in the training period. One of the advantages of the proposed methodology is the possibility to score events for such new links. 

A range of MEG models are fitted to the training data, using different combinations of $r$ and $d$ for characterising main effects and interactions. A good configuration for the initial parameter values is obtained through utilising the quantities $u_i=\tfrac{N_i(T)}{nT}$ and $u_j^\prime=\tfrac{N_j^\prime(T)}{nT}$, corresponding to the average rate of incoming and outgoing connections observed for each node. In particular, good results and convergence are obtained setting initial values $\alpha_i=\mu_i=u_i,\ \phi_i=3u_i$, $\beta_j=\mu_j^\prime=u_j^\prime$, and $\phi_j^\prime=3u_j^\prime$. For the interaction term, the initial values used to obtain the results are $\gamma_{i\ell}=\gamma_{j\ell}^\prime=\nu_{i\ell}=\nu_{j\ell}^\prime=10^{-4}$, and $\theta_{i\ell}=\theta_{j\ell}^\prime=5\cdot10^{-4}$. If $d>1$, then Gaussian noise with standard deviation $2\cdot10^{-5}$ is added to the interaction parameters. In general, the algorithm is fairly robust to different initialisations if the scale of the parameters is similar to the choices above. The learning rate $\eta$ is set to $0.1$.

Three strategies are used for estimation of $\tau_{ij}$: \begin{enumerate*}[label=(\roman*)] \item \label{item:mle} Using the MLE $\hat\tau_{ij}=t_{\ell_{ij1}}$; \item Setting $\tau_{ij}=0$; \item \label{item:aij} Setting $\tau_{ij}=0$ if $A_{ij}=1$, and $\tau_{ij}=\infty$ if $A_{ij}=0$ \end{enumerate*}. The MLE approach \ref{item:mle} has a drawback: the $p$-values \eqref{pval} for the first observation on each edge are \textit{always} 1. This implies that the KS scores are bounded below by ${2720}/{30704}\approx0.0885$ for the training set and $287/{3723}\approx 0.0770$ for the test set. 

The KS scores obtained on the training and test sets after fitting different MEG models are reported in Table~\ref{ks_enron}. 
The best performance (KS score $0.0152$) is achieved when a Markov process is used for the interaction term, with $d=5$ or $d=10$, combined with a Hawkes process for the main effects, setting $\tau_{ij}$ using option \ref{item:aij}. 
The same model achieves the best performance when alternative strategies for estimation of $\tau_{ij}$ are used. If $\tau_{ij}$ is set to its MLE \ref{item:mle}, then the lower bound for the KS score on the training set is attained. 
In general, setting $\tau_{ij}$ using option \ref{item:aij} seems to outperform competing strategies for estimation of $\tau_{ij}$ in terms of KS scores. 
More importantly, overall the results demonstrate that the interaction term plays a key role in obtaining a good fit on the observed event times. 

\begingroup

\begin{table}[!t]
\centering
\caption{Training and test Kolmogorov-Smirnov scores on the Enron e-mail network for different configurations of the MEG model.}
\label{ks_enron}
\scalebox{0.75}{
\begin{tabular}{c | c c | c c c c c}
\toprule
\multicolumn{3}{c|}{\textbf{KS scores (train \& test)}}  & \multicolumn{4}{c}{Main effects $\alpha_i(\cdot)$ and $\beta_j(\cdot)$ $\downarrow$} \\
\midrule
\multicolumn{1}{c|}{$\tau_{ij}$ $\downarrow$} & \multicolumn{2}{c|}{Interactions $\gamma_{ij}(\cdot)$ $\downarrow$} & Absent & Poisson ($r=0)$ & Markov ($r=1$) & Hawkes ($r=\infty$) \\
\midrule

\multirow{10}{*}{\parbox{4cm}{\centering $\tau_{ij}=t_{\ell_{ij1}}$ (MLE)}} & \multicolumn{2}{c|}{Absent} & -- \hspace{.8cm} -- & 0.4530 0.4133 & 0.3678 0.3484 & 0.4443 0.3586 \\
\cline{2-7}

& \multirow{3}{*}{\parbox{1.5cm}{\centering Poisson \\ ($r=0$)}} & $d=1$ & 0.4252 0.4221 & 0.3946 0.4179 & 0.3434 0.3574 & 0.4255 0.3560 \\
& & $d=5$ & 0.3490 0.3851 & 0.3498 0.3953 & 0.3165 0.3677 & 0.3491 0.3613 \\
& & $d=10$ & 0.3339 0.3763 & 0.3347 0.3688 & 0.3112 0.3470 & 0.3376 0.3575 \\
\cline{2-7}

& \multirow{3}{*}{\parbox{1.5cm}{\centering Markov \\ ($r=1$)}} & $d=1$ & 0.1662 0.2029 & 0.1491 0.1945 & 0.1305 0.1777 & 0.1702 0.1874 \\
& & $d=5$ & 0.0916 0.1875 & 0.0910 0.1684 & \textbf{0.0885 0.1628} & 0.0916 0.1746 \\
& & $d=10$ & \textbf{0.0885} 0.1743 & \textbf{0.0885} 0.1848 & \textbf{0.0885} 0.1696 & \textbf{0.0885} 0.1743 \\
\cline{2-7}

& \multirow{3}{*}{\parbox{1.5cm}{\centering Hawkes \\ ($r=\infty$)}} & $d=1$ & 0.2640 0.2755 & 0.2825 0.2887 & 0.2538 0.2637 & 0.2599 0.2871 \\
& & $d=5$ & 0.2304 0.2904 & 0.2284 0.2760 & 0.2271 0.2774 & 0.2420 0.2981 \\
& & $d=10$ & 0.2461 0.2923 & 0.2521 0.2865 & 0.2413 0.3091 & 0.2498 0.3129 \\

\midrule 

\multirow{10}{*}{\parbox{4cm}{\centering $\tau_{ij}=0$}} & \multicolumn{2}{c|}{Absent} & -- \hspace{.8cm} -- & 0.7678 0.7983 & 0.7456 0.7360 & 0.7058 0.6046 \\
\cline{2-7}

& \multirow{3}{*}{\parbox{1.5cm}{\centering Poisson \\ ($r=0$)}} & $d=1$ & 0.7039 0.7926 & 0.6627 0.7753 & 0.6543 0.7148 & 0.7059 0.6050 \\
& & $d=5$ & 0.5623 0.7059 & 0.5646 0.7206 & 0.5748 0.7008 & 0.7060 0.6053 \\
& & $d=10$ & 0.5354 0.6853 & 0.5332 0.6739 & 0.5725 0.6952 & 0.7060 0.6059 \\
\cline{2-7}

& \multirow{3}{*}{\parbox{1.5cm}{\centering Markov \\ ($r=1$)}} & $d=1$ & 0.3135 0.3324 & 0.3004 0.3326 & 0.3262 0.3240 & 0.2027 0.1999 \\
& & $d=5$ & 0.0760 0.1664 & 0.0825 0.1584 & 0.0855 0.1782 & 0.0495 \textbf{0.0924} \\
& & $d=10$ & 0.0775 0.1649 & 0.0793 0.1546 & 0.0816 0.1606 & \textbf{0.0402} 0.0971 \\
\cline{2-7}

& \multirow{3}{*}{\parbox{1.5cm}{\centering Hawkes \\ ($r=\infty$)}} & $d=1$ & 0.2871 0.2486 & 0.2333 0.2449 & 0.2485 0.2379 & 0.1749 0.1991 \\
& & $d=5$ & 0.1939 0.2167 & 0.1885 0.2246 & 0.2010 0.2137 & 0.1467 0.1994 \\
& & $d=10$ & 0.2029 0.2395 & 0.2158 0.2470 & 0.2207 0.2339 & 0.1606 0.1943 \\

\midrule 

\multirow{10}{*}{\parbox{4cm}{\centering $\displaystyle{\tau_{ij}=\begin{cases}0,& A_{ij}=1\\\infty,& A_{ij}=0\end{cases}}$}} & \multicolumn{2}{c|}{Absent} & -- \hspace{.8cm} -- & 0.5590 0.5941& 0.4112 0.3667 & 0.4593 0.2758 \\
\cline{2-7}

& \multirow{3}{*}{\parbox{1.5cm}{\centering Poisson \\ ($r=0$)}} & $d=1$ & 0.5158 0.6038 & 0.4812 0.5864 & 0.3742 0.3602 & 0.4197 0.2808 \\
& & $d=5$ & 0.4269 0.5516 & 0.4309 0.5641 & 0.3553 0.3598 & 0.3938 0.2803 \\
& & $d=10$ & 0.4035 0.5413 & 0.4084 0.5565 & 0.3430 0.3537 & 0.3659 0.2810 \\
\cline{2-7}

& \multirow{3}{*}{\parbox{1.5cm}{\centering Markov \\ ($r=1$)}} & $d=1$ & 0.1950 0.2115 & 0.1600 0.2017 & 0.1504 0.1422 & 0.1309 0.1445 \\
& & $d=5$ & 0.0709 0.1222 & 0.0746 0.1008 & 0.0696 0.0917 & \textbf{0.0152} 0.0848 \\
& & $d=10$ & 0.0619 0.1029 & 0.0627 0.1079 & 0.0634 0.0836 & 0.0213 \textbf{0.0800} \\
\cline{2-7}

& \multirow{3}{*}{\parbox{1.5cm}{\centering Hawkes \\ ($r=\infty$)}} & $d=1$ & 0.1870 0.2084 & 0.1816 0.2049 & 0.1783 0.1747 & 0.1719 0.1879 \\
& & $d=5$ & 0.1377 0.1805 & 0.1374 0.1840 & 0.1391 0.1642 & 0.1553 0.2154 \\
& & $d=10$ & 0.1556 0.2023 & 0.1588 0.2046 & 0.1546 0.1863 & 0.1640 0.2082 \\

\bottomrule

\end{tabular}
}
\end{table}
\endgroup

The results on the training set can also be compared to alternative node-based models from the literature. For example, \cite{Fox16} propose the following node-specific intensity function for sending e-mails:
\begin{equation}
\lambda_i(t) = \alpha_i+\sum_{k=1}^{N_i^\prime(t)} \mu_i\exp\{-(\mu_i+\phi_i)(t-t_{ik}^\prime)\}, \label{fox_model}
\end{equation}
where the intensity jumps according to the event times of the \textit{received} e-mails, \textit{cf.} \eqref{node_process} and \eqref{scaled_expo}. Despite the present article using a slightly different number of e-mails, the Kolmogorov-Smirnov score obtained on the training data using \eqref{fox_model} is $0.2806$, which corresponds almost exactly to the result in \cite{Fox16}, demonstrating that the MEG appears to have superior performance for the Enron network. The parameters of \eqref{fox_model} are estimated by direct optimisation using the Nelder-Mead method on the negative log-likelihood function for each source node. Nearly identical results to \cite{Fox16} are also obtained from fitting an independent Poisson processes $\lambda_i(t)=\alpha_i$ on each source node, with KS score $0.4088$. Finally, independent Hawkes process models of the form $\eqref{mep}$ are also fitted to each source node, obtaining a KS score of $0.2499$ which is significantly outperformed by the best configuration of the MEG model. Since the MEG model KS score outperforms the value obtained using \eqref{fox_model}, it could be inferred that users tend to respond to multiple e-mails in sessions, and not necessarily immediately after an individual e-mail is received.

\subsection{Imperial College London NetFlow data} \label{icl_data}

Many enterprises routinely collect network flow (NetFlow) data, representing summaries of connections between internet protocol (IP) addresses \citep[see, for example,][]{Hofstede14}, which should be monitored for detecting unusual network activity, security breaches, and potential intrusions. 
Modelling arrival times in computer networks is complicated by several factors: events tend to appear in bursts, they might be recorded multiple times, and exhibit polling at regular intervals \citep{Heard14}. In computer network security, it is particularly important to assess the significance of observing \textit{new links}, corresponding to connections on previously unobserved edges \citep{Metelli19}. New links might be indicative of lateral movement, which is a common behaviour of network attackers \citep{Neil13}: intruders might move across the network with the purpose of escalating credentials, establishing connections which were previously unseen or unexpected. Therefore, correctly modelling new connections, and consequently providing reliable anomaly scores, is paramount for 
network security. The proposed MEG framework for modelling point processes on networks simultaneously addresses two fundamental tasks in network security: 
monitoring the normality of observed traffic, and anomaly detection for unusual new connections. 

A bipartite dynamic network has been constructed from a subset of NetFlow data collected at Imperial College London (ICL). The network consists of $\numprint{1951067}$ arrival times recorded to the nearest millisecond, observed between 20th January 2020, and 9th February 2020, recorded from $n_1=173$ clients hosted within the Department of Mathematics at ICL, connecting to 
$n_2=\numprint{6083}$ internet servers connecting on ports 80 and 443 (corresponding to unencrypted and encrypted web traffic), forming a total of 
$\numprint{156186}$ unique edges. The periodic and automated activity has been filtered by considering only edges such that the percentage of arrival times observed between 7am and 12am is larger than 99\%, corresponding to the 
building opening hours. 
To learn connectivity patterns, the MEG model is trained on the first two weeks of data, corresponding to $\numprint{1299372}$ events, and tested on $\numprint{651695}$ events observed in the final week. The number of unique edges observed in the training period is $\numprint{115600}$, and $\numprint{70408}$ in the test set; only $\numprint{29822}$ edges are observed in both time windows, which implies that $\numprint{40586}$ new edges are observed in the test set. 

As discussed in Section~\ref{intro}, computer network data are observed in bursts and exhibit periodic behaviour. Figure~\ref{example_icl} gives an example of the connections from two of the clients to the ICL Virtual Learning Environment (VLE) server. Each session begins at an hour consistent with human behaviour, while the frequency of subsequent connections within each session is likely to be due to automated activity and page refreshing.

The models have been initialised using a similar initialisation scheme to Section~\ref{enron_data}, with learning rate $\eta=0.5$. In particular, setting $u_i=\tfrac{N_i(T)}{n_1 T}$ and $u_j^\prime=\tfrac{N_j^\prime(T)}{n_2 T}$, the chosen initial values are $\alpha_i=\mu_i=u_i,\ \phi_i=3u_i$, $\beta_j=\mu_j^\prime=u_j^\prime$, and $\phi_j^\prime=3u_j^\prime$, $\gamma_{i\ell}=(u_i)^{1/2}$, $\gamma_{j\ell}^\prime=(u_j^\prime)^{1/2}$, $\nu_{i\ell}=\nu_{j\ell}^\prime=10^{-4}$, and $\theta_{i\ell}=\theta_{j\ell}^\prime=5\cdot10^{-4}$. As before, Gaussian noise is added to the interaction parameters if $d>1$.

\begin{figure}[!t]
\centering
\begin{subfigure}[t]{.495\textwidth}
\centering
\includegraphics[width=\textwidth]{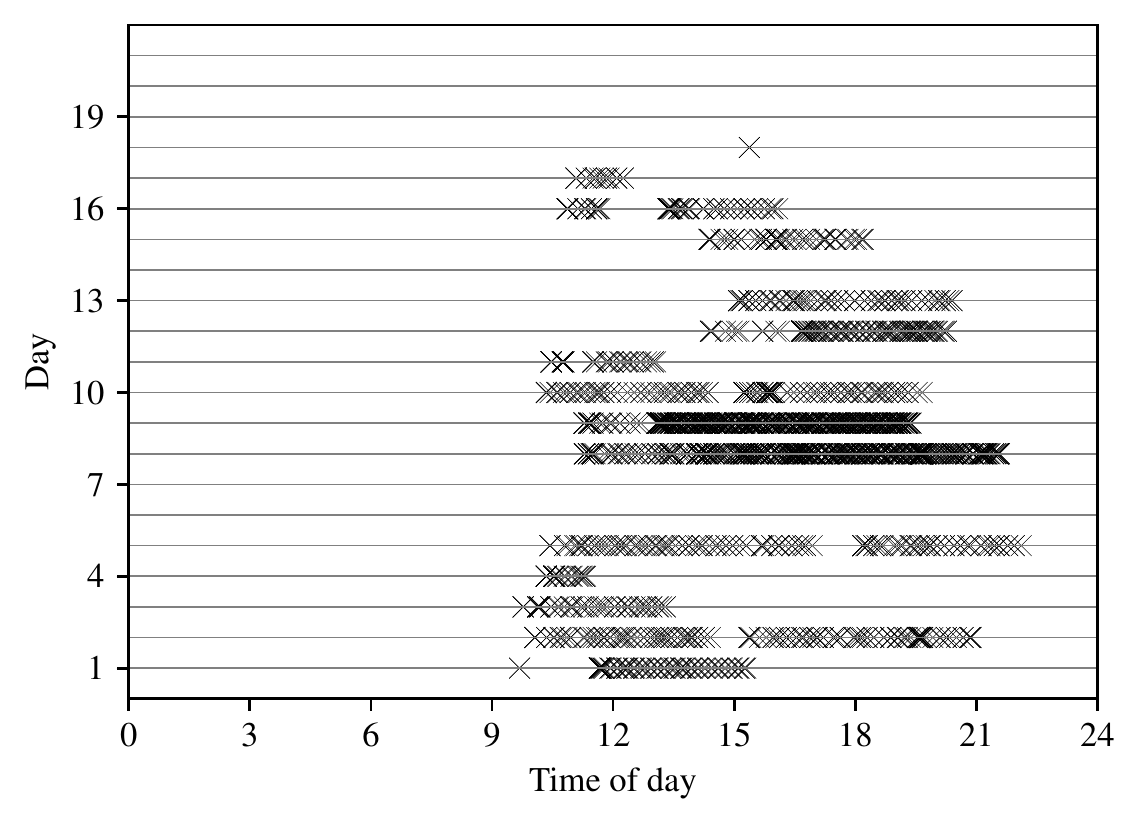}
\end{subfigure}
\begin{subfigure}[t]{.495\textwidth}
\centering
\includegraphics[width=\textwidth]{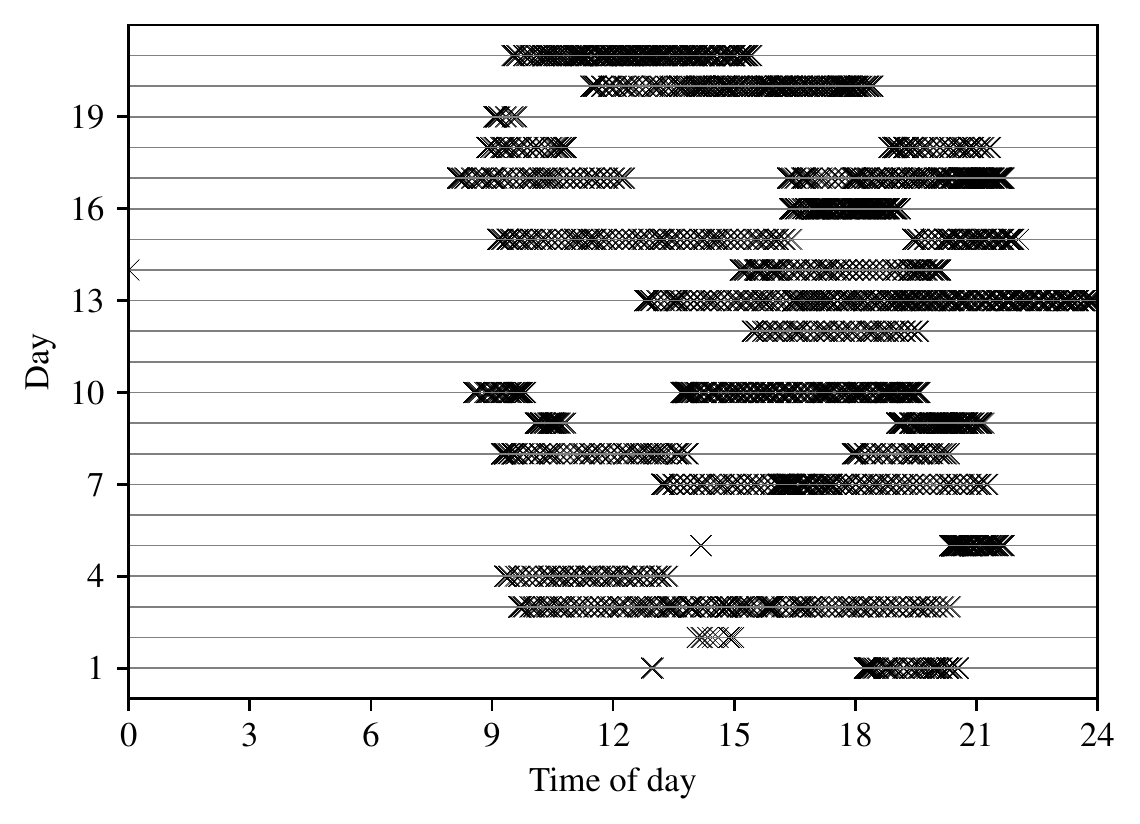}
\end{subfigure}
\caption{Connections to the ICL Virtual Learning Environment from two clients.}
\label{example_icl}
\end{figure}

The likelihood for the Hawkes process is highly multimodal, and more sensitive to the initial values of the parameters than the Markov process with $r=1$. Therefore, the parameters for the Hawkes process models are initialised with the optimal values obtained from the corresponding Markov process models, which seems to lead to fast convergence. The Kolmogorov-Smirnov scores calculated on the training and test set arrival times for different MEG models are reported in Table~\ref{ks_icl}. The parameter $\tau_{ij}$ is set according to option \ref{item:aij} from Section \ref{enron_data}, which was observed to have the best performance on the Enron data. 

The best performance (KS score $0.0728$) is achieved by a Markov process with $r=1$ for both the main effects and interactions, and latent dimensionality $d=5$ for the parameters of the interaction term. Corresponding Q-Q plots for some of the models are plotted in Figure~\ref{qu_icl}. Overall, the table and plots demonstrate that correctly modelling the arrival times requires inclusion within the model of an interaction term with a self-exciting component. Because of the extremely bursty behaviour of NetFlow arrival times, the Markov process model for main effects and interactions intuitively appears to be a suitable choice. 

\begingroup

\begin{table}[t]
\centering
\caption{KS scores on the ICL NetFlow data for different configurations of the MEG model.}
\label{ks_icl}
\scalebox{0.77}{
\begin{tabular}{c c | c c c c c}
\toprule
\multicolumn{2}{c|}{\textbf{KS scores (train \& test)}}  & \multicolumn{4}{c}{Main effects $\alpha_i(\cdot)$ and $\beta_j(\cdot)$ $\downarrow$} \\
\midrule
\multicolumn{2}{c|}{Interactions $\gamma_{ij}(\cdot)$ $\downarrow$} & Absent & Poisson ($r=0)$ & Markov ($r=1$) & Hawkes ($r=\infty$) \\
\midrule

\multicolumn{2}{c|}{Absent} &  -- \hspace{.8cm} -- & 0.7351 0.7148 & 0.6678 0.6489 & 0.7312 0.6950 \\
\midrule

\multirow{3}{*}{\parbox{1.5cm}{\centering Poisson \\ ($r=0$)}} & $d=1$ & 0.7328 0.7157 & 0.7325 0.7150 & 0.6672 0.6480 & 0.7316 0.6960 \\
& $d=5$ & 0.7295 0.7167 & 0.7313 0.7123 & 0.6673 0.6487 & 0.7275 0.6967 \\
& $d=10$ & 0.7260 0.7174 & 0.7289 0.7140 & 0.6680 0.6493 & 0.7270 0.6969 \\
\midrule

\multirow{3}{*}{\parbox{1.5cm}{\centering Markov \\ ($r=1$)}} & $d=1$ & 0.2194 0.1723 & 0.2242 0.1657 & 0.2038 0.1440 & 0.1645 0.1281 \\
& $d=5$ & 0.1024 0.1080 & 0.0896 0.0805 & \textbf{0.0728 0.0738} & 0.1041 0.0899 \\
& $d=10$ & 0.0843 0.0764 & 0.0871 0.0761 & 0.0850 0.0843 & 0.1100 0.0883 \\
\midrule

\multirow{3}{*}{\parbox{1.5cm}{\centering Hawkes \\ ($r=\infty$)}} & $d=1$ & 0.1080 0.0802 & 0.0747 0.1182 & 0.1082 0.0794 & 0.0884 0.1262 \\
& $d=5$ & 0.1576 0.1819 & 0.1532 0.2126 & 0.1677 0.2143 & 0.2307 0.2383 \\
& $d=10$ & 0.1584 0.1935 & 0.1546 0.2112 & 0.1619 0.2206 & 0.2388 0.2503 \\

\bottomrule
\end{tabular}
}
\end{table}
\endgroup

\begin{figure}[!t]
\centering
\begin{subfigure}[t]{.495\textwidth}
\centering
\caption{Training set}
\includegraphics[width=\textwidth]{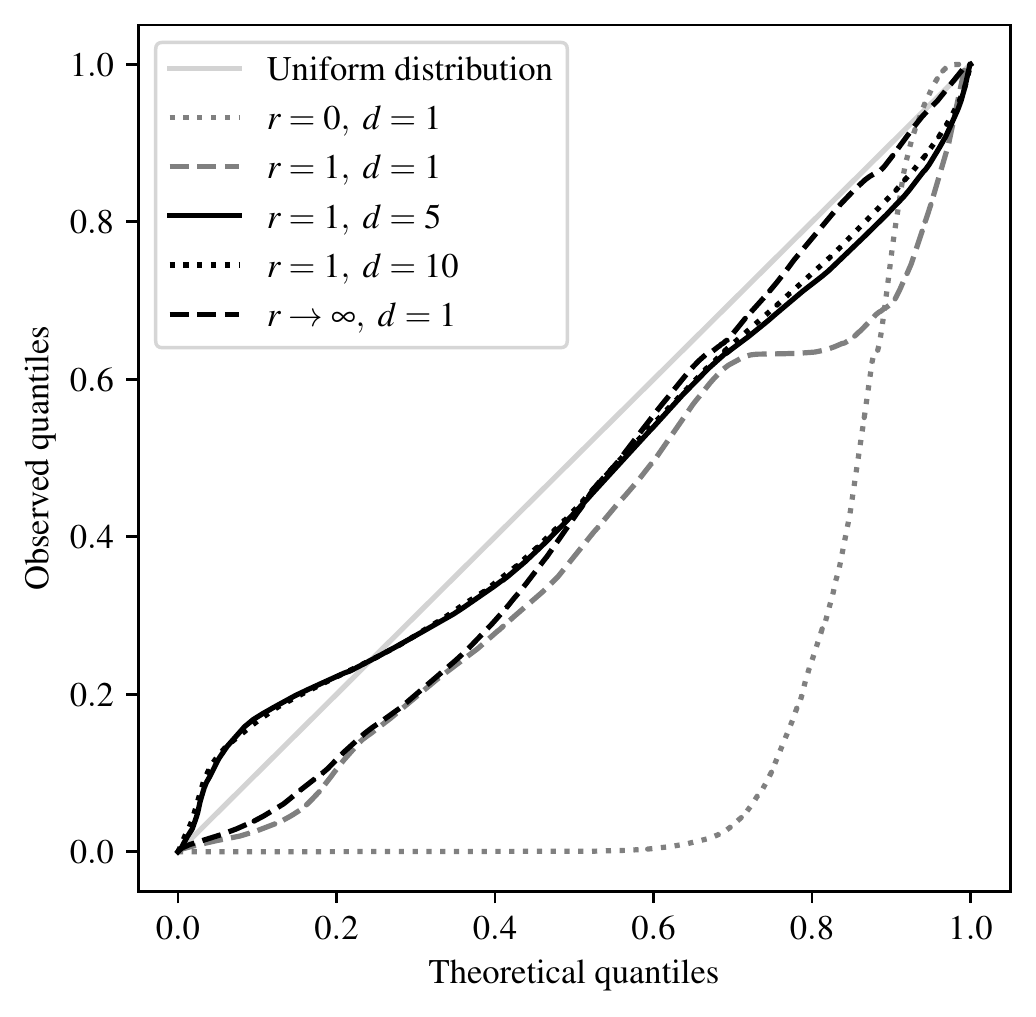}
\end{subfigure}
\begin{subfigure}[t]{.495\textwidth}
\centering
\caption{Test set}
\includegraphics[width=\textwidth]{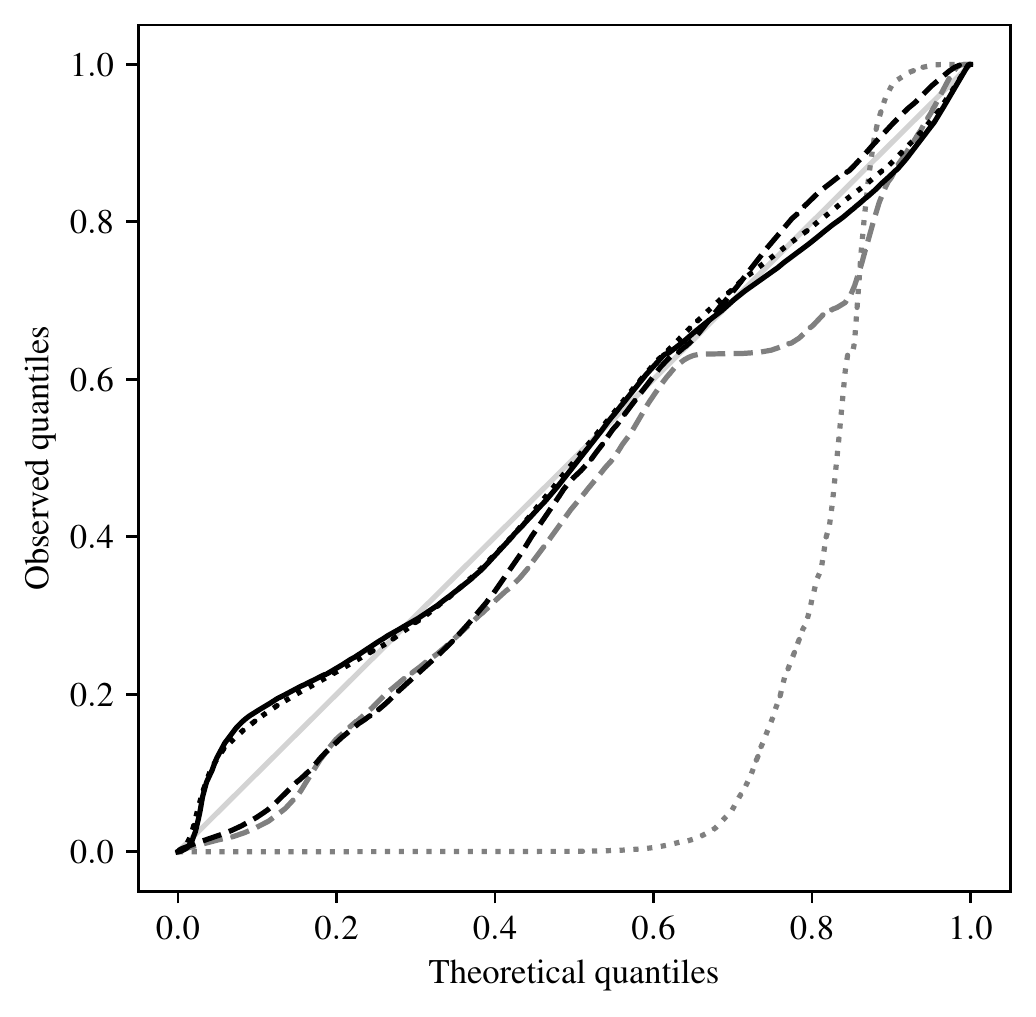}
\end{subfigure}
\caption{Q-Q plots for the training and test $p$-values obtained from different MEG models, with main effects $\alpha_i(t)$ and $\beta_j(t)$ with $r=1$, and different parameters for the interaction term $\gamma_{ij}(t)$, specified in the legend.}
\label{qu_icl}
\end{figure}

Finally, for the best performing model the corresponding KS scores are calculated individually for each edge, and plotted in Figure~\ref{ks_edge_icl} as a function of the number of connections on the edge. Clearly, the model has a better performance at scoring arrival times on more active edges. 

\begin{figure}[!t]
\centering
\begin{subfigure}[t]{.495\textwidth}
\centering
\caption{Training set}
\includegraphics[width=\textwidth]{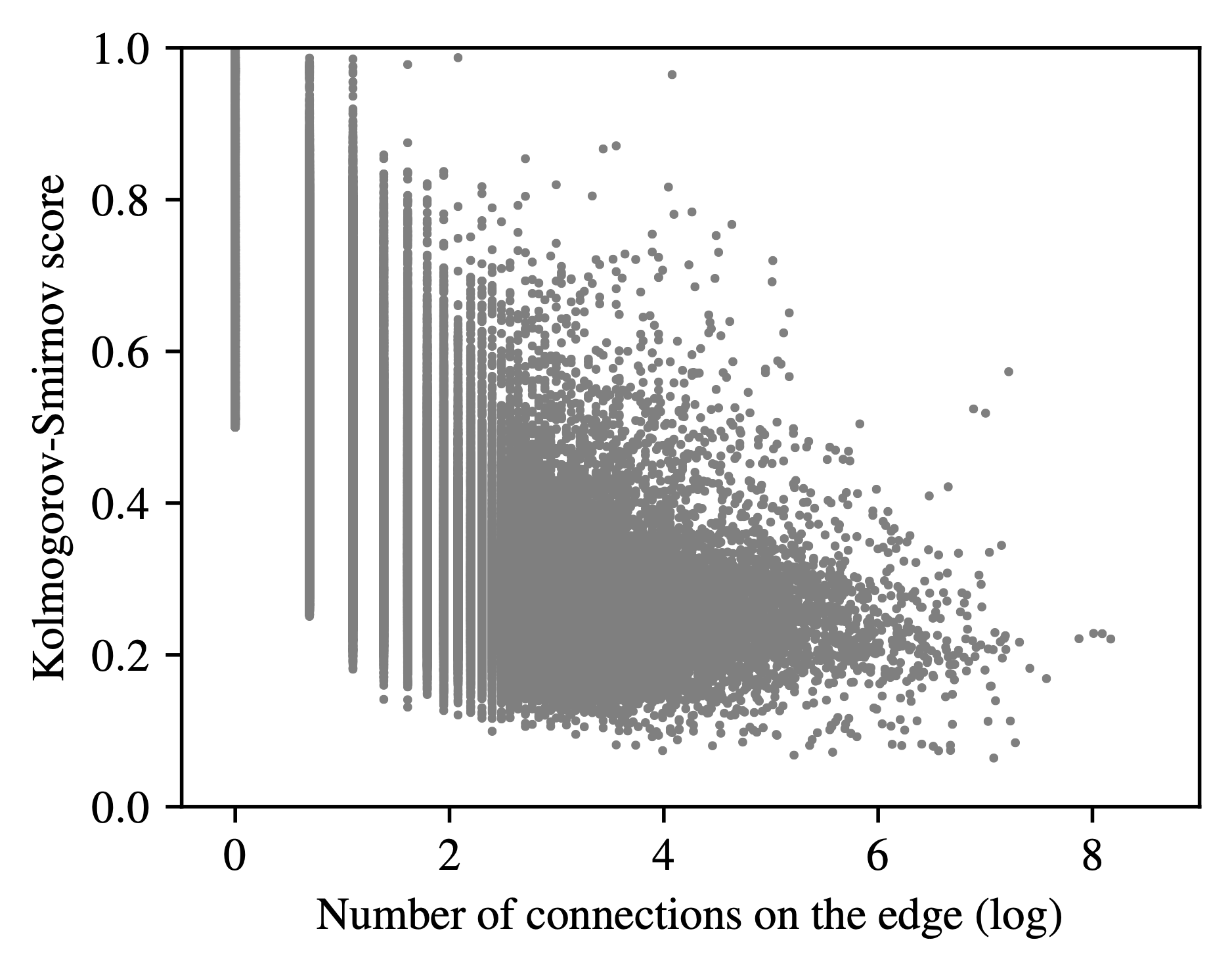}
\end{subfigure}
\begin{subfigure}[t]{.495\textwidth}
\centering
\caption{Test set}
\includegraphics[width=\textwidth]{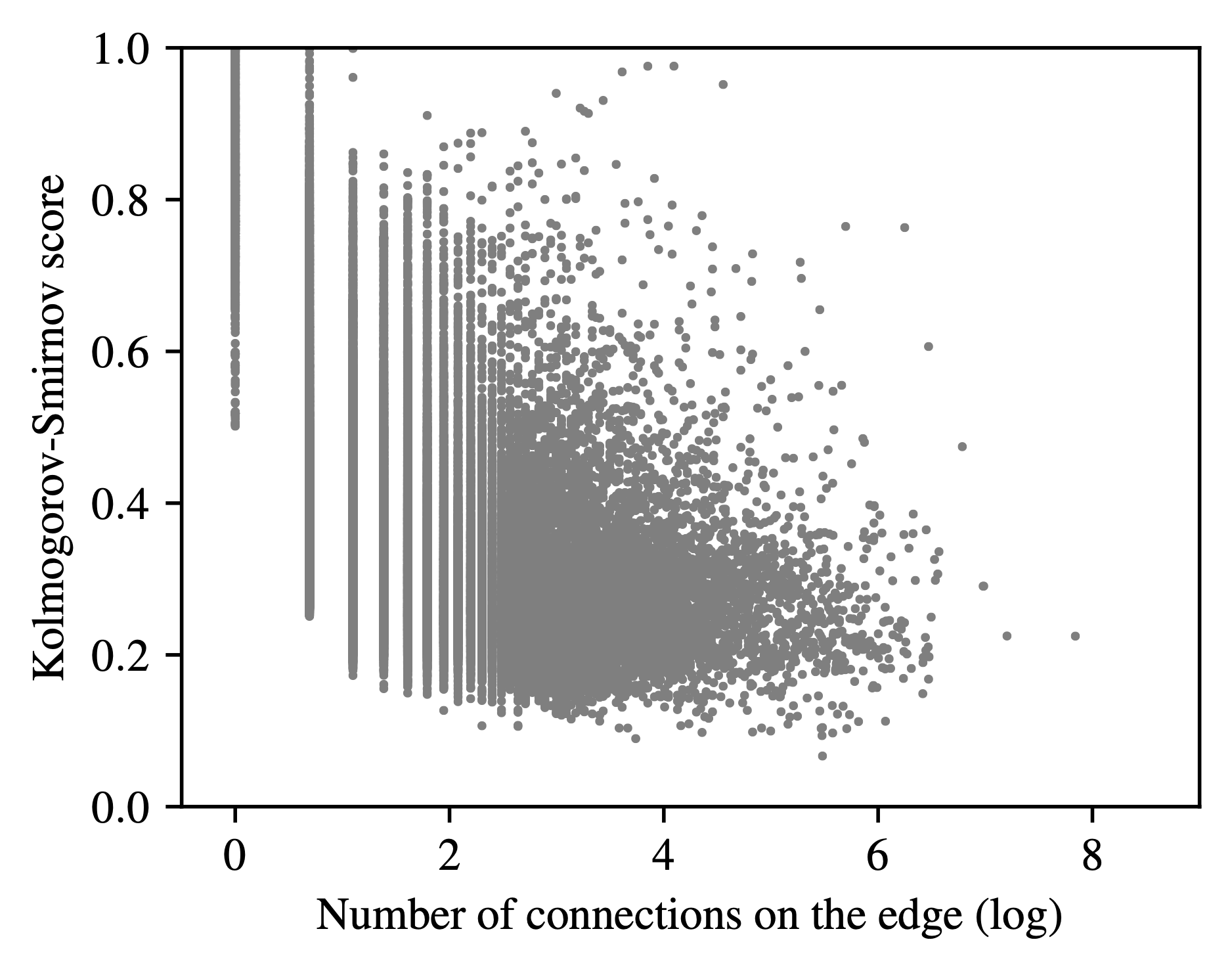}
\end{subfigure}
\caption{Scatterplot of the KS scores, calculated for each edge, versus the logarithm of the total number of connections on the edge, for the best performing model in Table~\ref{ks_icl}.}
\label{ks_edge_icl}
\end{figure}

\section{Conclusion} \label{sec:conclusion}

The mutually-exciting graph (MEG), a novel network-wide model for point processes with dyadic marks has been proposed. MEG uses mutually exciting point processes to model intensity functions, and borrows ideas from latent space models to infer relationships between the nodes. Edge-specific intensities are obtained only via node-specific parameters, which is useful for large and sparse graphs. Importantly, the proposed model is able to predict events observed on \textit{new} edges. 
Inference is performed via maximum likelihood estimation, optimised using the EM algorithm, or numerically via modern gradient ascent methods.
The model has been tested on simulated data and on two data sources related to computer networking: ICL NetFlow data and the Enron e-mail network. MEG appears to have excellent goodness-of-fit on training and testing data, resulting in low Kolmogorov-Smirnov scores even on very large and heterogeneous data like network flows. Furthermore, for the Enron e-mail network, MEG greatly outperforms results previously obtained in the literature on the same data.  
The model has been specifically motivated by cyber-security applications, where scoring observations on new links is particularly important for network security. Within this context, MEG might be used to complement existing techniques for modelling sequences of edges on dynamic networks \citep{SannaPassino19}, providing a network-wide method for scoring arrival times. 

The model could potentially be extended to admit an increasing number of nodes. Node-specific parameter values could be assigned after clustering similar nodes into groups, and allocating new nodes to one such community. For example, the initial parameter values for new nodes could correspond to the centroid of the corresponding community-specific parameter values. The initial cluster structure could be established from the connectivity patterns in the adjacency matrix (via spectral clustering or modularity maximisation), or from additional labels available for the nodes (for example, in cyber-security applications, subnets, geographical location, or machine type).  

\section*{Code}

A \textit{python} library to reproduce the results in this paper, and a \textit{bash} script to obtain the Enron e-mail network data, are available in the \textit{GitHub} repository \texttt{fraspass/meg}.

\section*{Acknowledgements}

This work is funded by the Microsoft Security AI research grant \textit{``Understanding the enterprise: Host-based event prediction for automatic defence in cyber-security"}. The authors thank Dr Melissa J. M. Turcotte for helpful discussions and comments.


\bibliographystyle{rss}
\singlespacing
\bibliography{biblio}

\appendix

\section{Calculation of the log-likelihood in MEG models} \label{likelihood_examples}


According to the choice of the excitation function and the parameter $r$, the log-likelihood \eqref{loglik} takes different forms. Here, examples are given for MEG models with $d=1$, $r=1$ or $r=\infty$, with scaled exponential excitation functions. 
For simplicity, since $d=1$, the second subscript $\ell$ is dropped from the triplets $(\gamma_{i\ell},\nu_{i\ell},\theta_{i\ell})$ and $(\gamma_{j\ell}^\prime,\nu_{j\ell}^\prime,\theta_{j\ell}^\prime)$.
Assume two sequences of arrival times $t_1<\cdots<t_{n_i}$ involving $i$ as source node, and $t_1^\prime<\cdots<t_{n_j^\prime}^\prime$ such that $j$ is the destination of the connection. Within the two sequences of arrival times $t_1,\dots,t_{n_i}$ and $t_1^\prime,\dots,t_{n_j^\prime}^\prime$, assume that a subset of $n_{ij}$ events, with $n_{ij}\leq\min\{n_i,n_j^\prime\}$, is observed on the edge $(i,j)$. Denote the indices of such events as $\ell_1,\dots,\ell_{n_{ij}}$ and $\ell_1^\prime,\dots,\ell_{n_{ij}}^\prime$, such that $t_{\ell_k}=t^\prime_{\ell_k^\prime}$. Define $\bar t_k=\max\{t_h:t_h<t_{\ell_k}\}$ and $\bar t_k^\prime=\max\{t_h^\prime:t_h^\prime<t_{\ell_k^\prime}^\prime\}$. For the edge $(i,j)$, assuming $r=1$, 
the first part of the log-likelihood \eqref{loglik} is:
\begin{align}
\sum_{k=1}^{n_{ij}} \log\lambda_{ij}(t_{\ell_k}) = \sum_{k=1}^{n_{ij}} \log\Big\{ & \alpha_i + \mu_i e^{-(\mu_i+\phi_i)(t_{\ell_k}-\bar t_k)} + \beta_j + \mu_j^\prime e^{-( \mu_j^\prime+\phi_j^\prime)(t_{\ell_k^\prime}^\prime-\bar t_k^\prime)} \\ & + \gamma_i\gamma_j^\prime + \nu_i\nu_j^\prime e^{-(\nu_i+\theta_i)(\nu_j^\prime+\theta_j^\prime)(t_{\ell_k}-t_{\ell_{k-1}})} \Big\}. \label{megw_lik}
\end{align}
For Hawkes process models, the main computational burden associated with the calculation of the likelihood is the double summation arising in the first term of \eqref{loglik} when $r=\infty$. The term in the first sum 
in \eqref{loglik} can be written as:
\begin{align}
\sum_{k=1}^{n_{ij}} \log\lambda_{ij}(t_{\ell_k})  = \sum_{k=1}^{n_{ij}} 
\log\bigg\{ & \alpha_i + \mu_i\sum_{h=1}^{\ell_k-1} e^{-(\mu_i+\phi_i)(t_{\ell_k}-t_h)} + \beta_j + \mu_j^\prime\sum_{h=1}^{\ell_k^\prime-1} e^{-(\mu_j^\prime+\phi_j^\prime)(t_{\ell_k^\prime}^\prime-t_h^\prime)} \\ & + \gamma_i\gamma_j^\prime +\nu_i\nu_j^\prime\sum_{h=1}^{k-1} e^{-(\nu_i+\theta_i)(\theta_j^\prime+\nu_j^\prime)(t_{\ell_k}-t_{\ell_h})}\bigg\}. \label{mega_likk}
\end{align}
Using a technique similar to the method proposed in \cite{Ogata78}, it is possible to calculate \eqref{mega_likk} in linear time using a recursive formulation of the inner summations. For $k\in\{2,\dots,n_{ij}\}$, define $\psi_{ij}(k)$, $\psi_{ij}^\prime(k)$ and $\tilde\psi_{ij}(t)$ as follows:
\begin{gather}
\psi_{ij}(k) = \sum_{h=1}^{\ell_k-1} e^{-(\mu_i+\phi_i)(t_{\ell_k}-t_h)}, \hspace{2cm} \psi_{ij}^\prime(k) = \sum_{h=1}^{\ell_k^\prime-1} e^{-(\mu_j^\prime+\phi_j^\prime)(t_{\ell_k^\prime}^\prime-t_h^\prime)}, \\ \tilde\psi_{ij}(k)=\sum_{h=1}^{k-1} e^{-(\nu_i+\theta_i)(\nu_j^\prime+\theta_j^\prime)(t_{\ell_k}-t_{\ell_h})}, \label{psi_def2}
\end{gather}
assuming the initial conditions $\tilde\psi_{ij}(1)=0$ and
\begin{align}
\psi_{ij}(1)=\sum_{h=1}^{\ell_1-1} e^{-(\mu_i+\phi_i)(t_{\ell_1}-t_h)}, && \psi_{ij}^\prime(1)=\sum_{h=1}^{\ell^\prime_1-1} e^{-(\mu_j^\prime+\phi_j^\prime)(t^\prime_{\ell^\prime_1}-t_h^\prime)}.
\end{align}
Note that the subscript $(i,j)$ for $\psi_{ij}(k)$ and $\psi_{ij}^\prime(k)$ is required since $\ell_k$ and $\ell_k^\prime$ are edge-specific values, and represent a short hand notation for $\ell_{ijk}$ and $\ell_{ijk}^\prime$. Using \eqref{psi_def}, the first term \eqref{mega_likk} of the likelihood becomes:
\begin{equation}
\sum_{k=1}^{n_{ij}} \log\lambda_{ij}(t_{\ell_k}) = \sum_{k=1}^{n_{ij}} \log\left\{\alpha_i+\beta_j+\gamma_i\gamma_j^\prime + \mu_i\psi_{ij}(k) + \mu_j^\prime\psi_{ij}^\prime(k) + \nu_i\nu_j^\prime\tilde\psi_{ij}(k) \right\}. \label{mega_lik2}
\end{equation}


\begin{proposition}
The terms $\psi_{ij}(k), \psi_{ij}^\prime(k)$ and $\tilde\psi_{ij}(k)$ can be written recursively as follows:
\begin{align}
& \psi_{ij}(k) = e^{-(\mu_i+\phi_i)(t_{\ell_k}-t_{\ell_{k-1}})}\left[1+\psi_{ij}(k-1)\right] + \sum_{h=\ell_{k-1}+1}^{\ell_k-1} e^{-(\mu_i+\phi_i)(t_{\ell_k}-t_h)}, \label{psi2}\\
& \psi_{ij}^\prime(k) = e^{-(\mu_j^\prime+\phi_j^\prime)(t^\prime_{\ell_k^\prime}-t^\prime_{\ell_{k-1}^\prime})}\left[1+\psi_{ij}^\prime(k-1)\right] + \sum_{h=\ell_{k-1}^\prime+1}^{\ell_k^\prime-1} e^{-(\mu_j^\prime+\phi_j^\prime)(t_{\ell_k^\prime}^\prime-t_h^\prime)},\\ 
& \tilde\psi_{ij}(k) = e^{-(\nu_i+\theta_i)(\nu_j^\prime+\theta_j^\prime)(t_{\ell_k}-t_{\ell_{k-1}})}\left[1+\tilde\psi_{ij}(k-1)\right].
\label{psi_prime2}
\end{align}
\end{proposition}
\begin{proof}
The result is proved here for $\psi_{ij}(k)$. 
\begin{align}
\psi_{ij}(k) &= \sum_{h=1}^{\ell_k-1} e^{-(\mu_i+\phi_i)(t_{\ell_k}-t_h)} = \sum_{h=1}^{\ell_{k-1}} e^{-(\mu_i+\phi_i)(t_{\ell_k}-t_h)} + \sum_{h=\ell_{k-1}+1}^{\ell_k-1} e^{-(\mu_i+\phi_i)(t_{\ell_k}-t_h)} \\
		&= e^{-(\mu_i+\phi_i)(t_{\ell_k}-t_{\ell_{k-1}})+(\mu_i+\phi_i)(t_{\ell_k}-t_{\ell_{k-1}})} \sum_{h=1}^{\ell_{k-1}} e^{-(\mu_i+\phi_i)(t_{\ell_k}-t_h)} + \sum_{h=\ell_{k-1}+1}^{\ell_k-1} e^{-(\mu_i+\phi_i)(t_{\ell_k}-t_h)} \\
		&= e^{-(\mu_i+\phi_i)(t_{\ell_k}-t_{\ell_{k-1}})}\sum_{h=1}^{\ell_{k-1}} e^{-(\mu_i+\phi_i)(t_{\ell_{k-1}}-t_h)} + \sum_{h=\ell_{k-1}+1}^{\ell_k-1} e^{-(\mu_i+\phi_i)(t_{\ell_k}-t_h)} \\
		&= e^{-(\mu_i+\phi_i)(t_{\ell_k}-t_{\ell_{k-1}})}\left[1+\sum_{h=1}^{\ell_{k-1}-1} e^{-(\mu_i+\phi_i)(t_{\ell_{k-1}}-t_h)}\right] + \sum_{h=\ell_{k-1}+1}^{\ell_k-1} e^{-(\mu_i+\phi_i)(t_{\ell_k}-t_h)},
\end{align}
which, using the definition of $\psi_{ij}(k-1)$, gives the result \eqref{psi2}. 
The proof for $\psi_{ij}^\prime(k)$ follows similar steps, but the summation is splitted at $\ell_{k-1}^\prime$, and the first summation is multiplied and divided by $\exp\{-(\mu_j^\prime+\phi_j^\prime)(t_{\ell_k^\prime}^\prime-t_{\ell_{k-1}^\prime}^\prime)\}$. For $\tilde\psi_{ij}(k)$, the decomposition is analogous to standard Hawkes processes. 
\end{proof}
The second part of the log-likelihood \eqref{loglik} is equivalent to $\Lambda_{ij}(T)$ and follows from integration of $\lambda_{ij}(t)$ over the observation period. For $r=1$: 
\begin{align}
&\int_0^T \lambda_{ij}(t)\mathrm dt =(\alpha_i+\beta_j+\gamma_i\gamma_j^\prime)
(T-\min\{T,\tau_{ij}\}) \\
& - \frac{\nu_i\nu_j^\prime}{(\nu_i+\theta_i)(\nu_j^\prime+\theta_j^\prime)}\sum_{k=1}^{n_{ij}}\left[e^{-(\nu_i+\theta_i)(\nu_j^\prime+\theta_j^\prime)(t_{\ell_{k+1}}-t_{\ell_k})}-1\right] \notag \\ & - \frac{\mu_i}{\mu_i+\phi_i}\sum_{k=1}^{n_i}\mathds 1_{[\tau_{ij},\infty)}(t_k)\left[e^{-(\mu_i+\phi_i)(t_{k+1}-t_k)}-1\right] \\
 & - \frac{\mu_j^\prime}{\mu_j^\prime+\phi_j^\prime}\sum_{k=1}^{n_j^\prime}\mathds 1_{[\tau_{ij},\infty)}(t_k^\prime)\left[e^{-(\mu_j^\prime+\phi_j^\prime)(t_{k+1}^\prime-t_k^\prime)}-1\right] \notag  \\ 
& - \frac{\mu_i}{\mu_i+\phi_i} \left[e^{-(\mu_i+\phi_i)(\min\{t_h:t_h\geq\tau_{ij}\}-\max\{t_h:t_h\leq\tau_{ij}\})}-e^{-(\mu_i+\phi_i)(\tau_{ij}-\max\{t_h:t_h\leq\tau_{ij}\})}\right] \notag \\ 
& - \frac{\mu_j^\prime}{\mu_j^\prime+\phi_j^\prime} \left[e^{-(\mu_j^\prime+\phi_j^\prime)(\min\{t_h^\prime:t_h^\prime\geq\tau_{ij}\}-\max\{t_h^\prime:t_h^\prime\leq\tau_{ij}\})}-e^{-(\mu_j^\prime+\phi_j^\prime)(\tau_{ij}-\max\{t_h^\prime:t_h^\prime\leq\tau_{ij}\})}\right], \label{comp_r1}
\end{align}
where $t_{n_i+1}=t_{n_j^\prime+1}^\prime=t_{\ell_{n_{ij}+1}}=T$.
\begingroup
\allowdisplaybreaks
Similarly, for $r=\infty$:  
\begin{align}
\int_0^T \lambda_{ij}(t)\mathrm dt = &\ (\alpha_i+\beta_j+\gamma_i\gamma_j^\prime)
(T-\min\{T,\tau_{ij}\}) \label{comp_rinf} \\
& - \frac{\mu_i}{\mu_i+\phi_i}\sum_{k=1}^{n_i}\left[e^{-(\mu_i+\phi_i)(T-t_k)}
- e^{-(\mu_i+\phi_i)\min\{T,\max\{\tau_{ij}-t_k,0\}\}}\right]  \\
& - \frac{\mu_j^\prime}{\mu_j^\prime+\phi_j^\prime}\sum_{k=1}^{n_j^\prime}\left[e^{-(\mu_j^\prime+\phi_j^\prime)(T-t_k^\prime)}
- e^{-(\mu_i+\phi_i)\min\{T,\max\{\tau_{ij}-t_k^\prime,0\}\}}\right] \\
&	- \frac{\nu_i\nu_j^\prime}{(\nu_i+\theta_i)(\nu_j^\prime+\theta_j^\prime)}\sum_{k=1}^{n_{ij}}\left[e^{-(\nu_i+\theta_i)(\nu_j^\prime+\theta_j^\prime)(T-t_{\ell_k})}-1\right]. \notag
\end{align}
\endgroup

Note that \eqref{comp_r1} and \eqref{comp_rinf} are monotonically decreasing functions in $\tau_{ij}$ for any choice of the remaining parameters, with constraint $\tau_{ij}<t_{\ell_{ij1}}$, where $t_{\ell_{ij1}}$ is the first arrival time on the edge $(i,j)$. Furthermore, $\tau_{ij}$ does not explicitly appear in the first part of the likelihood, \textit{cf.} \eqref{megw_lik} and \eqref{mega_likk}. Therefore, using \eqref{loglik}, the maximum likelihood estimate for $\tau_{ij}$ is simply $\hat\tau_{ij}=t_{\ell_{ij1}}$ if at least one event is observed on the edge, and $\hat\tau_{ij}=\infty$ otherwise. If $\tau_{ij}$ is set to its MLE, then the last two lines of \eqref{comp_r1} cancel out. 

For the $p$-values in \eqref{pval}, for $r=\infty$ and $\tau_{ij}=0$, the difference between the compensators 
is calculated sequentially at the observed times $t_1,\dots,t_{n_{ij}}$ using $\psi_{ij}(k), \psi_{ij}^\prime(k)$ and $\tilde\psi_{ij}(k)$:
\begin{align}
 \Lambda_{ij}(t_k) - \Lambda_{ij}(t_{k-1}) =\ & (\alpha_i+\beta_j+\gamma_i\gamma_j^\prime)(t_k-t_{k-1}) \\
& - \frac{\mu_i}{\mu_i+\phi_i}\left[\psi_{ij}(k)-N_i(t_k)-\psi_{ij}(k-1)+N_i(t_{k-1})\right] \\ 
& - \frac{\mu_j^\prime}{\mu_j^\prime+\phi_j^\prime}\left[\psi^\prime_{ij}(k)-N_j^\prime(t_k)-\psi^\prime_{ij}(k-1)+N_j^\prime(t_{k-1})\right] \\ 
	& - \frac{\nu_i\nu_j^\prime}{(\nu_i+\theta_i)(\nu_j^\prime+\theta_j^\prime)}\left[\tilde\psi_{ij}(k)-N_{ij}(t_k)-\tilde\psi_{ij}(k-1)+N_{ij}(t_{k-1})\right]. \notag
\end{align}
An expression similar to \eqref{comp_r1} can be used for 
$\Lambda_{ij}(t_k) - \Lambda_{ij}(t_{k-1})$ when $r=1$.
 
\section{Calculation of the gradient for $r=\infty$} \label{gradient_example}

In the derivations of the gradient, the following notation is used:
\begin{equation}
\chi_{ij}(k) = \alpha_i+\beta_j+\gamma_i\gamma_j^\prime + \mu_i\psi_{ij}(k) + \mu_j^\prime\psi_{ij}^\prime(k) + \nu_i\nu_j^\prime\tilde\psi_{ij}(k).
\end{equation}
The partial derivative of $\lolik$ with respect to $\alpha_i$ and $\gamma_i$ takes the following form: 
\begin{align}
&\der{\alpha_i}\lolik = \sum_{j=1}^n
\mathds 1_{[\tau_{ij},\infty)}(T) \left[-(T-\min\{T,\tau_{ij}\}) + \sum_{k=1}^{n_{ij}} \chi_{ij}(k)^{-1}\right], \\
&\der{\gamma_i}\lolik = \sum_{j=1}^n  
\gamma_j^\prime \mathds 1_{[\tau_{ij},\infty)}(T) \left[-(T-\min\{T,\tau_{ij}\}) + \sum_{k=1}^{n_{ij}} \chi_{ij}(k)^{-1}\right]. 
\end{align}
Similar equations can be derived for the partial derivatives with respect to $\beta_j$ and $\gamma_j^\prime$.

The calculations of the partial derivatives with respect to the parameters $\mu_i$, $\phi_i$, $\mu_j^\prime$ and $\phi_j^\prime$ use the recursive structure defined in the previous section, since the expressions $\psi_{ij}(k)$ and $\psi_{ij}^\prime(k)$ are functions of $(\mu_i,\phi_i)$ and $(\mu_j^\prime,\phi_j^\prime)$ respectively. For $\mu_i$ and $\phi_i$:
\begin{align}
\der{\mu_i}\lolik =&\ \frac{1}{\mu_i+\phi_i}\sum_{j=1}^n 
\mathds 1_{[\tau_{ij},\infty)}(T)
\sum_{k=1}^{n_i}\bigg\{
\frac{\phi_i}{\mu_i+\phi_i}\left[e^{-(\mu_i+\phi_i)(T-t_k)}-e^{-(\mu_i+\phi_i)\tau_{ijk}^+}\right] \notag \\ &
- \mu_i\left[(T-t_k)e^{-(\mu_i+\phi_i)(T-t_k)}-\tau_{ijk}^+e^{-(\mu_i+\phi_i)\tau_{ijk}^+}\right]
\bigg\} \notag \\ &
+ \sum_{j=1}^n
\mathds 1_{[\tau_{ij},\infty)}(T)
\sum_{k=1}^{n_{ij}} \frac{1}{\chi_{ij}(k)}\left[\psi_{ij}(k) + \mu_i\der{\mu_i}\psi_{ij}(k)\right], \notag \\
\der{\phi_i}\lolik =&\ - \frac{\mu_i}{\mu_i+\phi_i}\sum_{j=1}^n 
\mathds 1_{[\tau_{ij},\infty)}(T)
\sum_{k=1}^{n_i}\bigg\{ \frac{1}{\mu_i+\phi_i}\left[e^{-(\mu_i+\phi_i)(T-t_k)}-e^{-(\mu_i+\phi_i)\tau_{ijk}^+}\right] \notag \\ & 
+ \left[(T-t_k)e^{-(\mu_i+\phi_i)(T-t_k)}-\tau_{ijk}^+e^{-(\mu_i+\phi_i)\tau_{ijk}^+}\right]
\bigg\} \\ 
& + \mu_i\sum_{j=1}^n 
\mathds 1_{[\tau_{ij},\infty)}(T)
\sum_{k=1}^{n_{ij}} \frac{1}{\chi_{ij}(k)}\der{\phi_i}\psi_{ij}(k), \label{grad_phi}
\end{align}
where $\tau_{ijk}^+=\min\{T,\max\{\tau_{ij}-t_k,0\}\}$.
In the above expression, the partial derivative of $\psi_{ij}(k)$ with respect to $\mu_i$ and $\phi_i$ is computed recursively as follows:
\begin{align}
\der{\mu_i}\psi_{ij}(k) = \der{\phi_i}\psi_{ij}(k) = &\ e^{-(\mu_i+\phi_i)(t_{\ell_k}-t_{\ell_{k-1}})}\left\{\der{\phi_i}\psi_{ij}(k-1) -(t_{\ell_k}-t_{\ell_{k-1}})\left[1+\psi_{ij}(k-1)\right]\right\}  \notag \\ 
& - \sum_{h=\ell_{k-1}+1}^{\ell_k-1} (t_{\ell_k}-t_h)e^{-(\mu_i+\phi_i)(t_{\ell_k}-t_h)},
\end{align}
Similar considerations can be made for the partial derivatives with respect to 
$(\nu_i,\theta_i)$ and $(\nu_j^\prime,\theta_j^\prime)$. 
In this case, the recursive form stems from $\tilde\psi_{ij}(k)$, which is function of the two pairs of parameters. For $\nu_i$ and $\theta_i$:
\begin{align}
\der{\nu_i}&\lolik = \sum_{j=1}^n \frac{
\mathds 1_{[\tau_{ij},\infty)}(T)
\nu_j^\prime}{\nu_i+\theta_i} \sum_{k=1}^{n_{ij}}\Bigg\{\frac{\theta_i}{(\nu_i+\theta_i)(\nu_j^\prime+\theta_j^\prime)}\left[e^{-(\nu_i+\theta_i)(\nu_j^\prime+\theta_j^\prime)(T-t_{\ell_k})}-1\right] \\ & -\nu_i(T-t_{\ell_k})e^{-(\nu_i+\theta_i)(\nu_j^\prime+\theta_j^\prime)(T-t_{\ell_k})}\Bigg\} + \sum_{j=1}^n \sum_{k=1}^{n_{ij}} \frac{
\mathds 1_{[\tau_{ij},\infty)}(T)
\nu_j^\prime }{\chi_{ij}(k)}\left[\tilde\psi_{ij}(k) + \nu_i\der{\nu_i}\tilde\psi_{ij}(k)\right], \notag \\
\der{\theta_i}&\lolik = -\frac{\nu_i}{\nu_i+\theta_i}\sum_{j=1}^n 
\mathds 1_{[\tau_{ij},\infty)}(T)
\nu_j^\prime\sum_{k=1}^{n_{ij}}\Bigg\{\frac{1}{(\nu_i+\theta_i)(\nu_j^\prime+\theta_j^\prime)}\left[e^{-(\nu_i+\theta_i)(\nu_j^\prime+\theta_j^\prime)(T-t_{\ell_k})}-1\right] \notag \\ 
&+(T-t_{\ell_k})e^{-(\nu_i+\theta_i)(\nu_j^\prime+\theta_j^\prime)(T-t_{\ell_k})}\Bigg\} + \nu_i\sum_{j=1}^n \sum_{k=1}^{n_{ij}} \frac{1}{\chi_{ij}(k)}\der{\theta_i}\tilde\psi_{ij}(k).
\end{align}
The recursive equations for the partial derivative of $\tilde\psi_{ij}(k)$ with respect to $\nu_i$ and $\theta_i$ are equivalent. For $\theta_i$:
\begin{align}
\der{\theta_i}\tilde\psi_{ij}(k) = &\ e^{-(\nu_i+\theta_i)(\nu_j^\prime+\theta_j^\prime)(t_{\ell_k}-t_{\ell_{k-1}})}\left\{\der{\theta_i}\tilde\psi_{ij}(k-1)-(\nu_j^\prime+\theta_j^\prime)(t_{\ell_k}-t_{\ell_{k-1}})\left[1+\tilde\psi_{ij}(k-1)\right]\right\}, \notag
\end{align}
Similarly to the previous cases, the initial condition is:
\begin{equation}
\der{\theta_i}\tilde\psi_{ij}(1) 
= 0.
\end{equation}

\end{document}